\documentclass{article}

\usepackage{stmaryrd}
\usepackage{amsmath}
\usepackage{amssymb}
\usepackage{amsfonts}
\usepackage{graphicx}
\usepackage{subfigure}
\usepackage{enumerate}

\title{Local quantum field logic}

\author{{\sc Hector Freytes}}

\date{{\small
Universit\'a di Cagliari\\
Dipartimento di Matematica $\setminus$ Filosofia , \\ Via Is Mirionis I, Cagliari - Italia.\\
email: hfreytes@gmail.com}}

\begin{document}

\bibliographystyle{plain}

\maketitle

\begin{abstract}
\noindent
Algebraic quantum field theory, or AQFT for short, is a rigorous analysis
of the structure of relativistic quantum mechanics. It is formulated in
terms of a net of operator algebras indexed by regions of a Lorentzian
manifold. In several cases the mentioned net is represented by a family of
von Neumann algebras, concretely, type III factors. Local quantum field
logic arises as a logical system that captures the propositional structure
encoded in the algebras of the net. In this framework, this work contributes
to the solution of a family of open problems, emerged since the 30s, about
the characterization of those logical systems which can be identified with
the lattice of projectors arising from the Murray-von Neumann classification
of factors. More precisely, based on physical requirements formally described in AQFT, an equational theory
able to characterize the type III condition in a factor is provided.
This equational system motivates the study of
a variety of algebras having an underlying orthomodular lattice structure.
A Hilbert style calculus, algebraizable in the mentioned variety, is also introduced
and a corresponding completeness theorem is established.

\end{abstract}

\begin{small}



\end{small}

\bibliography{pom}

\begin{thebibliography}{10}

\bibitem{ARAKI} H. Araki, von Neumann algebras of local observables for free scalar field, Math. Phys. 5 (1964) 1-13.

\bibitem{BELT} E. Beltrametti, G. Cassinelli, The Logic of Quantum Mechanics, Encyclopedia of Mathematics and its Applications, N 15,  Cambridge University Press; Reissue edition 2010.


\bibitem{BvN} G. Birkhoff, and J. von Neumann,  The logic of quantum mechanics, Ann. Math. 27  (1936) 823-843.


\bibitem{BRO}  H.J. Borchers,  A remark on a theorem of B. Misra, Comm. Math. Phys. 4 (1967) 315-323.

\bibitem{BUN} L. J. Bunce, J. D. Maitland Wright, Quantum Logic, State Space Geometry and Operator Algebras, Comm. Math. Phys. 96 (1984) 345-348.


\bibitem{BrunsH} G. Bruns and J. Harding, Algebraic aspects of orthomodular lattices, in: B. Coecke, D. Moore, and A. Wilce (Eds),
Current Research in Operational Quantum Logic, Kluwer Academic Publishers (2000) pp. 37-65.

\bibitem{Bur} S. Burris,  H. P.  Sankappanavar,  A Course in Universal Algebra, Graduate Text in Mathematics, vol. 78, Springer-Verlag, New York Heidelberg Berlin, 1981.


\bibitem{DOP} S. Doplicher, R. Haag, J.E. Roberts, Local observables and particle statistics I, Comm. Math. Phys. 23 (1971) 199-230.


\bibitem{FILLMORE} P. Fillmore, Perspectivity in projection lattices, Proc. Amer. Math. Soc. 16 (1965) 383-387.


\bibitem{GRATZER} G. Gr\"atzer,  General Lattice Theory, 2nd edition, Bikh\"auser Verlag, Basel, Boston,  Berlin 1996.


\bibitem{HAGKAS} R. Haag, D. Kastler, An algebraic approach to quantum field theory, J. Math. Phys. 5 (1964) 848-861.

\bibitem{HAG} R. Haag, Local Quantum Physics: Fields, Particles, Algebras, 2nd edition, Springer, Berlin et al 1996.

\bibitem{HAL} P. Halmos,  Introduction to Hilbert space and the Theory of Spectral Multiplicity, 2nd edition, Dover Publications INC. Mineola, NY. 2018.



\bibitem{HOLAND} S. Holland, Distributivity and perspectivity in orthomodular lattices Trans. Am. Math. Soc. 112 (1964) 330-343.


\bibitem{HOLAND2} S. Holland, The Current Interest in Orthomodular Lattices, in: J. C. Abbott, (ed), Trends in Lattice Theory, Van Nostrand-Reinhold, New York (1970) pp. 41-26.

\bibitem{HOLAND3} S. Holland,  Orthomodularity in infinite dimensions; a theorem of M. Sol\'{e}r, Bull. Amer. Math. Soc. 32 (1995), 205-234.


\bibitem{HORRU} S. S. Horuzhy, Introduction to Algebraic Quantum Field Theorey, Kluwer Academic Publishers, Dordretch, Boston, London 1986.


\bibitem{Ka} J. A. Kalman, Lattices with involution, Trans. Amer. Math. Soc.  87 (1958) 485-491.



\bibitem{KAL}  G. Kalmbach,  Ortomodular Lattices, Academic Press, London, 1983.


\bibitem{KAPL} I. Kaplansky, Projections in Banach algebras, Ann. of Math. 53 (1951) 235-249.

\bibitem{KAPL2} I. Kaplansky, Rings of Operators, University of Chicago mimeographed notes 1955.


\bibitem{HALV} H. Halvorson , M. Mueger, Algebraic quantum field theory, in: J. Butterfield, J. Earman (eds), Handbook of the Philosophy of Physics. Kluwer Academic Publishers (2006) pp. 731-864.


\bibitem{LOOMIS} L. H. Loomis, The lattice theoretic background of the dimension theory of operator algebras, Mem. Amer. Math. Soc. 18, 1955.



\bibitem{MAEDA} F. Maeda, Relative dimensionality in operator rings, J. Sci. Hiroshima Univ. Ser. A-11 Math., (1941), 1-6.


\bibitem{MAEDA2} F. Maeda, Dimension functions on certain general lattices, J. Sci. Hiroshima Univ. Ser. A-1 Math. 19 (1955), 211-237.


\bibitem{MM} F. Maeda, S. Maeda, Theory of Symmetric Lattices, Springer-Verlag, Berlin, 1970.


\bibitem{MV} F.J. Murray, J. von Neumann, On rings of operators, Ann. of Math. 37 (1936) 116-229.


\bibitem{MV2} F.J. Murray and J. von Neumann. On rings of operators II, Trans. Amer. Math. Soc.  41 (1937) 208-248.


\bibitem{MV4} F.J. Murray and J. von Neumann, On rings of operators IV, Ann. of Math. 2 (1943) 716 - 808.


\bibitem{PENDERSEN} G. K. Pedersen, $C^*$-algebras and their Automorphism Groups, Academic Press, London 1979.


\bibitem{POWER} R. Power,  Existence of uncountable number of non isomorphic type $III$ factors,  Ann. of Math. 86  (1967) 138-171.


\bibitem{SASAKI} U. Sasaki,  Orthocomplemented lattices satisfying the exchange axiom, J. Sci. Hiroshima Univ. Ser. A-17, (1954) 293-302.



\bibitem{REDEI} M. R\'{e}dei, Why John von Neumann did not Like the Hilbert Space formalism of quantum mechanics, Studies in History and Philosophy of Modern Physics 27 (1996) 493-510.


\bibitem{vNlibro} J. von Neumann, Marhemarische Grundlagen der Quanrenmechanik, Heidelberg, Springer 1932. Transl. by R. Beyer (Princeton: Princeton University Press, 1955).

\bibitem{vonNeumann49} J. von Neumann,  On Rings of Operators. Reduction Theory, Ann. of Math, Second Series, 50 (1949) 401-485

\bibitem{vonNeumannIII} J. von Neumann,  On Rings of Operators. Reduction Theory III, Ann. of Math, Second Series, 41 (1940) 94-161.

\bibitem{vonNeumannGEO} J. von Neumann, Continuous Geometry Princeton Univ. Press, Princeton, (ed by I. Halperin), 1960.

\bibitem{YNG} J. Yngvason,  The role of type III factors in quantum field theory, Rep. Math. Phys. 55 (2005), 135-147.

\end{thebibliography}

\newtheorem{theo}{Theorem}[section]

\newtheorem{definition}[theo]{Definition}

\newtheorem{lem}[theo]{Lemma}

\newtheorem{met}[theo]{Method}

\newtheorem{prop}[theo]{Proposition}

\newtheorem{coro}[theo]{Corollary}

\newtheorem{exam}[theo]{Example}

\newtheorem{rema}[theo]{Remark}{\hspace*{4mm}}

\newtheorem{example}[theo]{Example}

\newcommand{\proof}{\noindent {\em Proof:\/}{\hspace*{4mm}}}

\newcommand{\skproof}{\noindent {\em  Sketch of proof:\/}{\hspace*{4mm}}}

\newcommand{\qed}{\hfill$\Box$}

\newcommand{\ninv}{\mathord{\sim}} %involutive negation

\newcommand\myeq{\mathrel{\stackrel{\makebox[0pt]{\mbox{\normalfont\tiny D}}}{\thickapprox}}}

\section*{Introduction}\label{INTRO}

Quantum field theory (QFT) is a set of tools that combines three areas of the modern physic: quantum theory, field theory and relativity. This theory underlies elementary particle physics and supplies essentials tools to other branches of the theoretical physics, such as condensed matter physics, statistical mechanics, astrophysics etc. Although quantum field theories have been developed and used for more than 70 years, a generally accepted rigorous description
of the structure of these theories has not yet been established. With the aim to establish a consistent mathematical framework for the treatment of  QFT, several axiomatic frameworks were formulated since the mid-fifties. One of these is known as {\it Algebraic quantum field theory} (AQFT).  Its origin lies in a seminal work of Haag and Kastler dating back to the early 1960s \cite{HAGKAS}. On this picture, a collection of observables is assigned to each open region of the Lorentzian spacetime. These observables correspond to physical quantities that can be measured by experiments causally confined to that region. The collection of observables comes equipped with an intrinsic operator algebra structure.  AQFT exists in two versions: the original Haag-Kastler formalism \cite{HAGKAS} based on $C^*$-algebras and the Haag-Araki formalism \cite{ARAKI, HAG} which uses von Neumann algebras. Here we adopt the second model mentioned above. Formally, the Haag-Araki model is based on a family $\{{\cal N}({\cal O})\}_{{\cal O}}$, called {\it net of local observable algebras over the spacetime}, where ${\cal O}$ is an open bounded region of a Lorentzian manifold  and ${\cal N}({\cal O})$ is a type III factor in the Murray-von Neumann classification of factors \cite{HALV, MV, MV2, MV4,vonNeumannIII}. Each algebra ${\cal N}({\cal O})$ mathematically represents the set of physical properties in the region ${\cal O}$ of the spacetime.  In this framework the aim of this work is to develop a logical system describing the propositional structure arising from the algebras of the net.

As is widely known, the elementary propositional structure associated to each algebra ${\cal N}({\cal O})$ is encoded in the orthomodular lattice ${\cal P}({\cal N}({\cal O}))$ defined by its projectors, each one of them, representing a true/false assertion related to a physical property in the region ${\cal O}$. However, the only orthomodularity condition on ${\cal P}({\cal N}({\cal O}))$ is not enough to distinguish the most important formal requirement of the algebras of the net, namely, the type III factor condition in the Murray-von Neumann classification.  {\it Local quantum field logic}, or $LQF$-logic for short, is an expansion of the  orthomodular logic that captures the type III factor condition of the algebras of the net. More precisely, this logical system is based on a necessary and sufficient condition formulated by a set of equations on an expanded language of the variety of orthomodular lattices that, when imposed on the projector lattice of a von Neumann factor, implies that this is a type III factor. In this perspective,  this work attempts to contribute to the solution of a family of open questions emerged since the 30's about whether it might be possible to establish lattice theoretical conditions in order to characterize each factor of the Murray-von Neumann classification \cite{BUN, HOLAND2, HOLAND3}.

The paper is structured as follows. Section \ref{BASICNOTION} contains generalities on universal algebra and lattice theory. Some technical results about operations expanding the orthomodular structure are also given. In Section \ref{VNlatticesDimension} an outline about operator algebras is provided. Moreover, useful facts linking elements of the Murray-von Neumann dimension theory and von Neumann lattices are established. Section \ref{RINGOP} provides a detailed motivation of the $LQF$-logic. In order to do this the physical framework underlying AQFT is briefly described.  In Section \ref{EQTYPEIII} the action of the partial isometries on the  lattice of projectors of a type III factor is studied. This allows us to transplant the Murray-von Neumann equivalence of type $III$ factors in the language of the lattice of projectors. Furthermore, partial isometries define a natural expansion of the language of the orthomdular lattices in which, a set of equations characterizing the type III factor, is formulated. Based on this equational theory, in Section \ref{ALMODELQF}, a variety of algebras, called $LQF$-algebras, is introduced and studied. Via $LQF$-algebras, alternative proofs of the crucial characteristics regarding to the non atomicity and non modularity of the  type III factors are given. Section \ref{FILTERCONGRUENCE} is devoted to the study of the congruences and filters in $LQF$-algebras. Finaly, in Section \ref{HILBERTSTYLE}, a Hilbert style calculus for $LFQ$-logic is introduced and a completeness theorem for this calculus is also established.

\section{Basic Notions}\label{BASICNOTION}

We first recall from \cite{Bur} some notion of universal algebra
that plays an important role along this article. Let $\tau$ be a type of algebras. We denote by $Term_{\tau}(X)$ the {\it absolutely free algebra} of type $\tau$
built  from the denumerable set of variables $X = \{x_1, x_2,...\}$. Each element of $Term_{\tau}(X)$ is referred to as a {\it ${\tau}$-term}. We denote
by $Comp(t)$ the complexity of the term $t$. An {\it equation of type $\tau$} is an expression of the form $t = s$ where $t,s \in Term_{\tau}(X)$.
For $t\in Term_{\tau}(X)$ we usually write $t(x_1, \ldots, x_n)$ to indicate that the variables occurring in $t$ are among $x_1, \ldots, x_n$.
A {\it variety} is a class of algebras of the same type defined by a set of equations. Let ${\cal A}$ be a variety of algebras of type $\tau$. In this case, $Term_{\tau}(X)$ is also denoted by
$Term_{\cal A}(X)$ and each element of $Term_{\tau}(X)$ is indistinctly referred to as a {\it ${\tau}$-term}, {\it ${\cal A}$-term} or simply {\it term} when there is no confusion.
A {\it ${\cal A}$-homomorphism} is a $\tau$-operation preserving map between two algebras of ${\cal A}$. If $A\in {\cal A}$ then we denote by $1_A$ the identity ${\cal A}$-homomorphism on $A$.
If ${\cal B}$ be a subclass of ${\cal A}$ then we denote by ${\cal V}({\cal B})$ the subvariety of ${\cal A}$
generated by the class ${\cal B}$, i.e. ${\cal V}({\cal B})$ is the
smallest subvariety of  ${\cal A}$ containing ${\cal B}$.
Let $A \in {\cal A}$. Each term $t(x_1, \ldots, x_n)$ in $Term_{\tau}(X)$ canonically defines an $n$-ary operation on $A$ denoted by $t^A$.
If $a_1,\dots, a_n \in A$ then we denote by $t^A(a_1,\dots, a_n)$
the result of the application of the term operation $t^A$ to the
elements $a_1,\dots, a_n$. A {\it valuation} in the algebra $A$ is a function of the form
$v:X\rightarrow A$. By induction on $Comp(t)$ any valuation $v$ in $A$ can be
uniquely extended to an ${\cal A}$-homomorphism $v:Term_{\tau}(X) \rightarrow A$, that is, $v(t(t_1, \ldots, t_n)) = t^A(v(t_1), \ldots,
v(t_n))$ for $t_1, \ldots, t_n \in Term_{\tau}(X)$.  Thus, valuations are identified with ${\cal A}$-homomorphisms from the absolutely free algebra. If $t,s \in
Term_{\tau}(X)$, $A \models t = s$ means that for each valuation $v$
in $A$, $v(t) = v(s)$ and ${\cal A}\models t=s$ means that for
each $A\in {\cal A}$, $A \models t = s$.

For each algebra $A \in {\cal A}$, we denote by $Con(A)$ the
congruence lattice of $A$, the diagonal congruence is denoted by
$\Delta_A$ and the largest congruence $A^2$ is denoted by $\nabla_A$.
A congruence $\theta$ is called  {\it factor congruence} iff there is a
congruence $\theta^*$ on $A$ such that, $\theta \land \theta^* =
\Delta_A$, $\theta \lor \theta^* = \nabla_A$ and $\theta$ permutes with
$\theta^*$. In this case the pair $(\theta, \theta^*)$ is called a {\it pair of factor
congruences} on $A$ and we can prove that  $A \cong A/\theta \times A/\theta^*$. The algebra $A$ is
{\it directly indecomposable} iff $A$ is not isomorphic to a product
of two non trivial algebras or, equivalently, if $\Delta_A,\nabla_A$ are
the only factor congruences in $A$. If ${\cal A}$ is a variety then we denote by ${\cal DI}({\cal A})$ the class of directly indecomposable algebras of ${\cal A}$.
An algebra $A$ has the {\it congruence extension property} (CEP) iff for each subalgebra $B$ and $\theta
\in Con(B)$ there is a $\phi \in Con(A)$ such that $\theta = \phi
\cap A^2$.  A variety ${\cal A}$ satisfies CEP iff every algebra
in ${\cal V}$ has the CEP.

The variety ${\cal A}$ is said to be {\it congruence distributive} iff for each $A \in {\cal A}$, $Con(A)$ is a distributive lattice. If for each $A \in {\cal A}$ the congruences of $Con(A)$ are permutable then we said that ${\cal A}$ is a {\it congruence permutable} variety. The variety ${\cal A}$ is {\it arithmetical} iff it is both congruence distributive and congruence permutable variety.

Let $A$ be an algebra. We say that $A$ is {\it subdirect product} of a family of $(A_i)_{i\in I}$ of algebras if
there exists an embedding $f: A \rightarrow \prod_{i\in I} A_i$ such
that $\pi_i f : A\! \rightarrow A_i$ is a surjective homomorphism
for each $i\in I$ where $\pi_i$ is the $i$th-projection onto $A_i$. The algebra $A$ is
{\it subdirectly irreducible} iff it is trivial or there is a
minimum congruence in $Con(A) - \Delta_A$. We denote by ${\cal SI}({\cal A})$ the class of subdirectly irreducible algebras of the variety ${\cal A}$.
It is clear that a subdirectly irreducible algebra is directly indecomposable. Then, for each variety ${\cal A}$, we have that
\begin{equation}\label{SIincDI}
{\cal SI}({\cal A}) \subseteq {\cal DI}({\cal A}).
\end{equation}

 An important result by Birkhoff is the following subdirect representation theorem.

\begin{theo}\label{Birkhoff} {\rm \cite[Theorem 8.6]{Bur}}
Let ${\cal A}$ be a variety. Then every algebra $A \in {\cal A}$ is a subdirect product of subdirectly irreducible algebras of ${\cal A}$.
\qed
\end{theo}

Let us notice that, by the above theorem and Eq.(\ref{SIincDI}), in each variety ${\cal A}$ the class ${\cal SI}({\cal A})$ and the class ${\cal DI}({\cal A})$ rule the valid equations in ${\cal A}$. That is, for any pair of terms $t, s\in Term_{\cal A}(X)$ we have that
\begin{equation}\label{EQDISI}
{\cal A}\models t=s \hspace{0.4cm} \mbox{iff} \hspace{0.4cm} {\cal DI}({\cal A})\models t=s \hspace{0.4cm} \mbox{iff} \hspace{0.4cm} {\cal SI}({\cal A})\models t=s.
\end{equation}

An algebra $A$ is said to be {\it simple} iff $Con(A) = \{\Delta_A, \nabla_A \}$. The class of simple algebras in the variety ${\cal A}$ is denoted by ${\cal S}im({\cal A})$.  The algebra $A$ is {\it semisimple} iff $A$ is a subdirect product of simple algebras. A variety ${\cal A}$ is semisimple iff each algebra of ${\cal A}$ is semisimple.  A {\it discriminator term} for the algebra $A$ is a term $t(x,y,z)$ such that
$$
t^A(x,y,z) = \begin{cases}x, & x\not=y, \\
z, & x=y. \end{cases}
$$
A variety ${\cal A}$ is a {\it discriminator variety} iff there exists a subclass of algebras ${\cal K}$ with a common discriminator term $t(x,y,z)$ such that ${\cal A} = {\cal V}({\cal K})$.

\begin{theo} \label{BULMAN} {\rm \cite[Theorem 9.4]{Bur}} Let ${\cal K}$ be a class of algebras of type $\tau$ and $t(x,y,z)$ be a common discriminator $\tau$-term for the class ${\cal K}$.
If we consider the generated variety ${\cal A} ={\cal V}({\cal K})$ then

\begin{enumerate}
\item
${\cal A}$ is an arithmetical semisimple variety.

\item
${\cal DI}({\cal A}) = {\cal SI}({\cal A}) = {\cal S}im({\cal A})$.

\end{enumerate}
\qed
\end{theo}

Now we recall from \cite{KAL} and \cite{MM} some notion about orthomodular lattices.  A {\it lattice with involution} \cite{Ka} is an algebra $\langle L, \lor, \land, \neg
\rangle$ such that $\langle L, \lor, \land \rangle$ is a  lattice
and $\neg$ is a unary operation on $L$ that fulfills the following
conditions: $\neg \neg x = x$ and  $\neg (x \lor y) = \neg x \land \neg y$.  An {\it orthomodular lattice} is an algebra $\langle L,
\land, \lor, \neg, 0,1 \rangle$ of type $\langle 2,2,1,0,0 \rangle$
that satisfies the following conditions:

\begin{enumerate}
\item
$\langle L, \land, \lor, \neg, 0,1 \rangle$ is a bounded lattice with involution,

\item
$x\land  \neg x = 0 $,

\item
$x\lor ( \neg x \land (x\lor y)) = x\lor y $.  \hspace*{\fill} {{\rm ({\it orthomodular law})} }

\end{enumerate}

We denote by ${\cal OML}$ the variety of orthomodular lattices. Upon defining the ${\cal OML}$-term $t R s$ as
\begin{equation}\label{EQCONNECTIVE}
t R s \hspace{0.2 cm} = \hspace{0.2 cm} (t \land s) \lor (\neg t\land \neg s)
\end{equation}
an important characterization of the equations in ${\cal OML}$ is given by:
\begin{equation}\label{ECMOL}
{\cal OML} \models t = s \hspace{0.4cm} iff  \hspace{0.4cm}
{\cal OML} \models t R s = 1.
\end{equation}
Therefore we can safely assume that all ${\cal OML}$-equations are of the form $t = 1$, where $t \in Term_{\cal OML}(X)$.
\begin{rema}\label{EQOML1}
{\rm  It is clear that the equational characterization given in Eq.(\ref{ECMOL}) is satisfied for each
variety ${\cal A}$ admitting terms of the language that define, on each $A\in {\cal A}$, operations $\lor$,
$\land$, $\neg$, $0,1$ such that $\langle A,\lor,\land, \neg, 0,1\rangle$ is an orthomodular lattice.
}
\end{rema}

Let $L$ be an orthomodular lattice. An element $a\in L$ is an {\it atom} iff $a\not = 0$ but $x\leq a$ implies, $x= 0$ or $x=a$. Dually, the notion of {\it coatom} is established. Let us notice that $a$ is an atom iff $\neg a$ is a coatom. Two  elements $a,b$ in $L$ are {\it orthogonal} (noted $a \bot b$) iff $a\leq \neg b$.  For each $a\in L$ let us consider the interval $[0,a] = \{x\in L : 0\leq x \leq a \}$ and the unary operation on  $[0,a]$ given by $\neg_a x = \neg x \land a$. As one can readly realize, the structure
\begin{equation}\label{internalalg}
[0,a]_L  = \langle [0,a], \land, \lor, \neg_a, 0, a \rangle
\end{equation}
is an orthomodular lattice.
Let $a\in L$. Then the mapping $\mu_a: L \rightarrow L$ given by
\begin{equation}\label{SASAKIPROJECTION}
\mu_a(x) = a\land (\neg a \lor x)
\end{equation}
is called the {\it Sasaki projection} onto $[0,a]$. In {\rm \cite[p. 156]{KAL}} it is proved that
\begin{equation}\label{SASAKIPROP}
x= \mu_a(x) \hspace{0.3cm} \mbox{iff}  \hspace{0.3cm} x\leq a.
\end{equation}

For elements $a,b \in L$ we said that $a$ {\it commutes} with $b$, in symbols $aCb$, iff $a = (a\lor b) \land (a\lor \neg b)$. It is not very hard to see that $aCb$ iff $bCa$. We also note that $aCb$, $\neg aCb$, $aC \neg b$, $\neg a C \neg b$ are equivalent conditions in an orthomodular lattices. \\

{\it Boolean algebras} are orthomodular lattices satisfying  the
{\it distributive law} $x\land (y \lor z) = (x \land y) \lor (x
\land z)$. We denote by ${\bf 2}$ the Boolean algebra of two
elements. Let $L$ be an orthomodular lattice. Given $a, b, c$ in $L$, we write: $(a,b,c)D$\ \   iff $(a\lor
b)\land c = (a\land c)\lor (b\land c)$; $(a,b,c)D^{*}$ iff $(a\land
b)\lor c = (a\lor c)\land (b\lor c)$ and $(a,b,c)T$\ \ iff
$(a,b,c)D$, (a,b,c)$D^{*}$ hold for all permutations of $a, b, c$.
An element $z$ of $L$ is called {\it central} iff for all elements
$a,b\in L$ we have\ $(a,b,z)T$. We denote by $Z(L)$ the set of all
central elements of $L$ and it is called the {\it center} of $L$.

\begin{prop}\label{eqcentro} Let $L$ be an orthomodular lattice. Then we have:

\begin{enumerate}

\item
$Z(L)$ is a Boolean sublattice of $L$ {\rm \cite[Theorem 4.15]{MM}}.

\item
$z \in Z(L)$ iff for each $a\in L$, $a = (a\land z) \lor (a \land \neg z)$  {\rm \cite[Lemma 29.9]{MM}}.
\end{enumerate}
\qed
\end{prop}

Let $L$ be an orthomodular lattice and $a\in L$. One can define the {\it central cover} of $a$, as
\begin{equation}\label{CENTRALCOVER}
e(a) = \bigwedge\{z\in Z(L): a\leq z\}
\end{equation}
if such infimum exists. By straightforward calculation we can see that
\begin{equation}\label{DUALCENTRALCOVER}
e_d(a) = \neg e(\neg a) = \bigvee\{z\in Z(L): z\leq a\}.
\end{equation}
if such supremum exists. This element is called the {\it dual central cover} of $a$. For example, if $L$ is a complete orthomodular then $e(a)$ and $e_d(a)$ exist for each $a\in L$. Moreover, they are central elements of $L$ {\rm \cite[Lemma 29.16]{MM}}.

\begin{prop}\label{DIOML}
Let $L$ be an orthomodular lattice and $z \in Z(L)$. Then:

\begin{enumerate}
\item
The binary
relation ${\theta}_z$ on $L$ defined by $a  \theta_z b$  iff $a\land z = b\land z$ is a congruence of $L$.

\item
The function $f_z: L/{\theta}_z \rightarrow [0,z]_L$ such that $f_z([a]_z) = a\land z $ is a ${\cal OML}$-isomorphism.

\item
$L$ is ${\cal OML}$-isomorphic to $L/{\theta}_z\times L/{\theta}_{\neg z}$ i.e., $(\theta_z, \theta_{\neg z})$ is a pair of factor congruences on $L$.

\item
The map $z \rightarrow \theta_z$ is a lattice isomorphism between
$Z(L)$ and the Boolean subalgebra of $Con_{{\cal OML}}(L)$ of factor congruences.

\item
$L$ is directly indecomposable iff $Z(L) = \{0,1\}$.

\end{enumerate}
\qed
\end{prop}

For a proof of Proposition \ref{DIOML} we refer to  {\rm \cite[\S 4]{BrunsH}}. Let us notice that the above proposition say that in an orthomodular lattice there exists a one to one correspondence between its central elements and the factor congruences. In other words, the presence of non trivial central elements in an orthomodular lattice determines each possible decomposition of the lattice in a direct product of two orthomodular lattices.

\begin{prop}\label{BAAZ}
Let $L$ be a directly indecomposable orthomodular lattice. Then the operation
$$
w_0(x) = \begin{cases}0, & x = 0, \\
1, & otherwise. \end{cases}
$$
is the unique operation on $L$ that satisfies the following conditions:
\begin{equation}\label{w00}
w_0(0) = 0,
\end{equation}
\begin{equation}\label{w0mX}
X \leq w_0(X),
\end{equation}
\begin{equation}\label{w0BOOL}
Y  = (Y\land w_0(X)) \lor (Y\land w_0(X)^\bot).
\end{equation}
\end{prop}

\begin{proof}
We first note that $w_0$ satisfy  Eq.(\ref{w00}), Eq.(\ref{w0mX}) and Eq.(\ref{w0BOOL}). Let $v$ be an operation on $L$ satisfying the mentioned equations.  Combining Eq.(\ref{w0BOOL}), Proposition \ref{eqcentro}-2 and Proposition \ref{DIOML}-5 we can see that $Imag(v) \subseteq Z({\cal P}({\cal N})) = \{0,1_{{\cal H}}\}$. Then, by Eq.(\ref{w0mX}), $v(x) = 1$ iff $x\not = 0$. Hence $v=w_0$.

\qed
\end{proof}

Let $L$ be an orthomodular lattice. For elements $a,b \in L$ we say that $a$ is a {\it complement} of $b$ iff $a\lor b = 1$ and $a\land b = 0$.

\begin{prop}\label{COMPLEMENTS} {\rm \cite[Proposition 6]{KAL}}.
Let $L$ be an orthomodular lattice and $a\in L$. Then the complements of $a$ are precisely the elements of the image of the following operation
\begin{equation}\label{COMPEQ}
L\ni x \mapsto c_a(x) = \big(x\land \neg(x\land a)\big)\lor \neg(x\lor a).
\end{equation}
\end{prop}

Let $L$ be an orthomodular lattice and  $a,b \in L$. We say that {\it $a$ and $b$ are perspective} and we write $a \sim_p b$ iff they have a common complement, i.e. if there exists $c\in L$ such that $a\lor c = 1 = b\lor c$ and $a\land c = 0 = b\land c$.

\begin{prop}\label{EQPESPEC}
Let $L$ be an orthomodular lattice and $a,b \in L$. Then the following statement are equivalent:

\begin{enumerate}

\item
There exists $x\in L$ such that $a\lor x = b\lor x$ and $a\land x = b\land x$.

\item
$a \sim_p b$.

\end{enumerate}
\end{prop}

\begin{proof}
If we assume that $a\lor x = b\lor x$ and $a\land x = b\land x$ then, by Proposition \ref{COMPLEMENTS}, $c_a(x) = (x\land \neg (x\land a) ) \lor \neg(x\lor a) = (x\land \neg (x\land b) ) \lor \neg(x\lor b) = c_b(x)$ is a common complement of $a$ and $b$. The other direction is immediate.

\qed
\end{proof}

\begin{prop}\label{PERSPECTCENTER}
Let $L$ be an orthomodular lattice and $a,b \in L$ such that $a \sim_p b$. Then,

\begin{enumerate}
\item
For each $z\in Z(L)$, $a\land z \sim_p b \land z $.

\item
If $z\in Z(L)$ and $z \leq a$ then $z\leq b$.

\item
$e_d(a) = e_d(b)$ whenever such elements exist in $L$.

\end{enumerate}

\end{prop}

\begin{proof}
Let us assume that $a \sim_p b$ in $L$.

1) Let $c \in L$ be a common complement of $a$ and $b$ and $z\in C(L)$. Note that $(a\land z) \land c  = 0 = a\land (c\land z)$ and $(a\land z) \lor c = (a\lor c) \land (z\lor c)  = (z\lor c) = (b\lor c) \land (z\lor c)$. Then, by Proposition \ref{EQPESPEC}, $a\land z \sim_p b\land z $.

2) By item 1 we have that  $z \sim_p b \land z $ because $z = a\land z$. Since $z$ is a central element then, necessarily,  $\neg z$ is the common complement of $z$ and $b \land z$. Thus, $1= \neg z \lor (b \land z) = \neg z \lor b$ and $z = z\land 1 = z \land (\neg z \lor b) = z\land b$. Hence $z\leq b$.

3) By item 1 we have that $e_d(a) = a\land  e_d(a) \sim_p b \land  e_d(a) $. Since $e_d(a)\in Z(L)$ then $\neg e_d(a)$ is the common complement of $e_d(a)$ and $b \land  e_d(a)$. Thus, $b \land  e_d(a) \in Z(L)$ and
$b \land  e_d(a) = e_d (b \land  e_d(a) ) = e_d (b) \land  e_d(a)$. In this way, $1= \neg e_d(a) \lor \big(e_d (b) \land  e_d(a) \big) = \neg e_d(a) \lor e_d (b)$ proving that $e_d(a) \leq e_d(b)$. Similarly we can prove that $e_d(b) \leq e_d(a)$. Hence, $e_d(a) = e_d(b)$.

\qed
\end{proof}

Let $L$ be an orthomodular lattice and $a,b \in L$. We say that $(a,b)$ is a {\it modular pair}, in symbols $(a,b)M$, iff for every $x \leq b$, $(x\lor a) \land b = x \lor (a\land b)$. By substituting $x$ with $x \land b$, $(a,b)M$ is equivalent to the following equation:
\begin{equation}\label{MODULARITYEQ}
(x\land b) \lor (a\land b) = ((x\land b) \lor a)\land b.
\end{equation}
The lattice $L$ is called {\it modular} iff $(a,b)M$ holds for all elements $a$ and $b$ of $L$. In other words, $L$ is modular iff it satisfies Eq.(\ref{MODULARITYEQ}). We can show that if $a\bot b$ then $(a,b)M$ holds. This property motivated Kaplansky to introduce the name ``orthomodular" as an abbreviation for ``orthogonal pairs are modular". In order to characterize the modularity we introduce the following lattice known as  $N_5$.

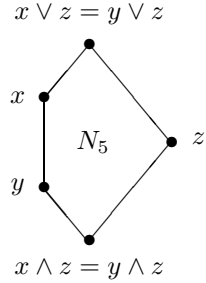
\begin{center}
\unitlength=1mm
\begin{picture}(60,20)(-5,7)
\put(18,-6.5){\line(-5,6){6}}
\put(18,20){\line(-5,-6){6}}
\put(18,-6){\line(5,6){11}}
\put(18,20){\line(5,-6){11}}

\put(12,1){\line(0,1){12}}

\put(18,-6){\circle*{1.5}}
\put(18,20){\circle*{1.5}}

\put(12,13){\circle*{1.5}}
\put(12,1){\circle*{1.5}}

\put(29,7){\circle*{1.5}}

\put(20.5,24){\makebox(-5,0){$x\lor z = y \lor z$}}
\put(35,7.5){\makebox(-5,0){$z$}}
\put(20.5,-10){\makebox(-5,0){$x\land z = y \land z$}}

\put(10,13){\makebox(-3,0){$ x$}}
\put(10,1){\makebox(-3,0){$ y$}}
\put(20,7){\makebox(-3,0){$N_5$}}

\end{picture}
\end{center}

\vspace{1.5cm}

\begin{prop}\label{EQUIVNONMOD}
Let $L$ be an orthomodular lattice. Then the following  statements are equivalent:

\begin{enumerate}
\item
$L$ is non modular.

\item
There exists $x,y \in L$ such that $x<y$ and $x\sim_p y$.

\item
$N_5$ is a sublattice $L$.

\end{enumerate}
\end{prop}

\begin{proof}
$1 \Longleftrightarrow 2$) Immediate. $1 \Longleftrightarrow 3$) See {\rm \cite[\S2, Theorem 2]{GRATZER}}.

\qed
\end{proof}

The technical results provided by the following two propositions turn out to be useful in the rest of the paper.

\begin{prop}\label{SASAKINRG}
Let $L$ be an orthomodular lattice, $a\in L$ such that $a\not=0$ and let us consider a function $w:[0,a]_L \rightarrow L$. Then the following statement are equivalent:
\begin{enumerate}
\item
$w(\mu_a(\neg x)) = \neg w(x\land a)$.

\item
For each $x\in [0,a]$, $w(\neg_a x) = \neg w(x)$ i.e. $w$ preserves relatives complements related to $a$.

\end{enumerate}
\end{prop}

\begin{proof}
$1 \Longrightarrow 2)$ Let us suppose that $x\leq a$. Since $\neg a\leq \neg x$ then we have that $\neg w(x) = \neg w(x\land a) = w(\mu_a(\neg x)) =  w(a \land (\neg a \lor \neg x)) = w(a \land \neg x) = w(\neg_a x)$.

$2 \Longrightarrow 1)$
Let $x\in L$. Since $x\land a \in [0,a]$ then $\neg w(x \land a) = w(\neg_a (x \land a)) = w(a \land \neg(x\land a)) = w(a \land ( \neg a \lor \neg x)) = w(\mu_a(\neg x))$.

\qed
\end{proof}

\begin{prop}\label{EXTEN}
Let $L$ be an orthomodular lattice and $a\in L$ such that $a\not=0$. Suppose that there exist two unary order preserving operations $w_a$ and $w_a^*$ on $L$ such that $w_a w_a^* = id_L$ and  $w_a^* w_a = \mu_a$. Then,

\begin{enumerate}
\item
$w_a^*:L \rightarrow [0,a]_L$ is an order isomorphism.

\item
The restriction $w_a\upharpoonright_{[0,a]_L} \rightarrow L$ is an order isomorphism.

\item
$w_a^*(a) = a$ iff $a=1$.

\end{enumerate}

\end{prop}

\begin{proof}
1) We first prove that $Imag(w_a^*) = [0,a]$. Let $x \in L$. Then $w_a^*(x) =  w_a^* \big( w_a w_a^*(x) \big) = \big( w_a^* w_a \big) w_a^*(x) = \mu_a(w_a^*(x)) \in [0,a]$. Let $y \in [0,a]_L$. By Eq.(\ref{SASAKIPROP}) and by hypothesis we have that $y = \mu_a(y) = w_a^* w_a(y)$ and then, $w_a(y)$ is a preimage of $y$ by $w_a^*$. Thus,  $Imag(w_a^*) = [0,a]$. Let $x,y \in L$ and let us assume that $w_a^*(x) = w_a^*(y)$. Then, by hypothesis, $x =  w_aw_a^*(x) = w_aw_a^*(y) = y$. Therefore,  $w_a^*$ is injective.  Since $w_a^*$ is an order preserving map, by the above conditions, $w_a^*:L \rightarrow [0,a]_L$ is an order isomorphism.

2) If $x\not = y$ in $[0,a]_L$ then, by hypothesis, $w_a^* w_a(x) = \mu_a(x) = x  \not = y = \mu_a(y) = w_a^* w_a(y)$. Thus, by item 1, $w_a(x) \not = w_a(y)$.  Therefore, the restriction $w_a\upharpoonright_{[0,a]_L}$ is an injective map. Let $y \in L$. If $x = w_a^*(y)$ then $w_a(x) = w_a w_a^*(y) = y$. Thus  $w_a$ is a surjective map. Since $w_a$ is an order preserving map, by the above conditions, the restriction $w_a\upharpoonright_{[0,a]_L} \rightarrow L$ is an order isomorphism.

3) Let us suppose that $w_a^*(a) = a$. Then $a= w_a w_a^*(a) = w_a(a)$ and, by item 2, $a= 1$. For the converse, if we suppose that $a = 1$ then, by item 1, $w_1^*(1) = 1$.

\qed
\end{proof}

\section{Murray-von Neumann dimension theory and von Neumann lattices} \label{VNlatticesDimension}

In this section we study some general properties about the lattice of projectors of a von Neumann algebra arising from the Murray-von Neumann classification of factors. In order to do this we first summarize some basic notions about the Murray-von Neumann dimension theory.

Let ${\cal H}$ be a Hilbert space with inner product $\langle \cdot, \cdot \rangle $ and norm $\Vert \cdot \Vert $. Two elements $x,y$ are said to be {\it orthogonal}  iff $\langle x, y\rangle = 0$ which can be written as $x\bot y$. For a set  $X \subseteq {\cal H}$ the {\it orthogonal complement } of $X$ is given by  $$X^\bot = \{y\in {\cal H}: \forall x \in X, \hspace{0.2cm} \langle x, y\rangle = 0 \}.$$  Two subspaces $X, Y$ are {\it orthogonal}, in symbols $X \bot Y$, iff $X \subseteq Y^\bot$ (equivalently $X^\bot \subseteq Y$).  If $X, Y$ are subspaces of ${\cal H}$ then the {\it sum of $X$ and $Y$} is given by the subspace $X+Y = \{ x+y : x\in X, y\in Y\}$.  In particular, ${\cal H}$ is said to be the {\it direct sum} of $X$ and $Y$ iff ${\cal H} = X + Y$ and $X\cap Y = \{0\}$. In this case we write ${\cal H} = X \oplus Y$. Let us remark that ${\cal H} = X \oplus Y$ iff any $a \in {\cal H}$  has a unique decomposition $a = a_{_X} + a_{_Y}$ where $a_{_X} \in X$ and $a_{_Y} \in Y$.  The following proposition provides well known results on Hilbert spaces (see {\rm \cite[\S 12]{HAL}}).

\begin{prop}\label{closed1}
Let ${\cal H}$ be a Hilbert space and $X, Y$ subspaces of ${\cal H}$. Then,

\begin{enumerate}
\item
$X^\bot$ is a closed subspace of ${\cal H}$.

\item
If $X$ is closed then ${\cal H} = X \oplus X^\bot$.

\item
If $X$ and $Y$ are closed and $X \bot Y$ then $X+Y$ is the closed subspace of ${\cal H}$ generated by $X \cup Y$.

\end{enumerate}

\qed
\end{prop}

An indexed subset $\{e_i\}_i$ of ${\cal H}$ is said to be {\it orthonormal} iff, $\Vert e_i \Vert = 1$ for all $i$ and $e_i\bot e_j$ for all $i\not = j$.
The space  ${\cal H}$ is {\it separable} iff  it has a countable basis, i.e., there exists a denumerable orthonormal set $\{e_i\}_{i=1}^\infty$ such that for each $x\in {\cal H}$, $x = \sum_{i=1}^\infty \langle x, e_i \rangle e_i$. \\

{\it From here on, we confine ourselves to separable Hilbert spaces}.\\

A linear operator $L: {\cal H} \rightarrow {\cal H}$, or operator for short, is said to be {\it bounded} iff $\sup_{\Vert x \Vert \leq 1}\{\Vert L(x) \Vert\} < \infty$. If $(e_i)_{i\in I}$ is a basis of ${\cal H}$ where $I \subseteq \mathbb{N}$ then the {\it trace} of $L$ is defined as $tr(L) = \sum_{i \in I } \langle L(e_i),e_i \rangle$. We can prove that the trace of $L$ is independent of the choice of basis. The {\it kernel} of $L$ is defined as the set
$Ker(L) = \{x\in {\cal H}: L(x)=0 \}$. Let us remark that $Ker(L)$ and the image of $L$, denoted by $Imag(L)$, are subspaces of ${\cal H}$. The set of all bounded operators on ${\cal H}$ is denoted by ${\cal B}({\cal H})$. The identity operator on ${\cal H}$ is denoted by $1_{\cal H}$ and immediately follows that it is a bounded operator. For each operator $L \in {\cal B}({\cal H})$ we denote by $L^*$ the {\it adjoint of $L$} i.e., the unique operator in ${\cal B}({\cal H})$ such that $\langle L(x), y \rangle = \langle x, L^*(y) \rangle$ for each $x,y \in {\cal H}$.

The relationship between  $Imag(L)$ and $Ker(L^*)$ is given by the following proposition.

\begin{prop}\label{KERRANG}
Let $L \in {\cal B}({\cal H})$. Then,
\begin{enumerate}
\item
$Ker(L) = Imag^\bot(L^*)$.

\item
$Ker^\bot(L) = \overline{Imag(L^*)}$ i.e. the closure of $Imag(L^*)$.

\end{enumerate}

\qed
\end{prop}

\begin{rema}\label{CLOSEDKERIMAG}
{\rm Combining Proposition \ref{closed1} and Proposition \ref{KERRANG} we see that for each $L \in {\cal B}({\cal H})$, $Ker(L)$ and $Ker^\bot(L)$ are closed subspaces of ${\cal H}$.
}
\end{rema}

The set ${\cal B}({\cal H})$ can be endowed with a linear space structure over $\mathbb{C}$ by considering the sum operation $(L+M)(x) = L(x) + M(x)$ and a product by complex scalars defined by $(\lambda L)(x) = \lambda L(x)$. In ${\cal B}({\cal H})$ we can also define the product operation $L\cdot M$ given by the composition $(L\cdot M)(x) = L(M(x))$. For the sake of simplicity the product operation $L\cdot M$ is written $L M$.

An operator $P$ on ${\cal H}$ is an {\it orthogonal projection} or simply a {\it projector} iff $P = P^* = P^2$. Note that if $P$ is a projector on ${\cal H}$, then $Imag(P)$ is a closed subspace of ${\cal H}$. Reciprocally, if $E$ is a closed subspace of ${\cal H}$ then there exists a projector $P_E$ such that $Imag(P_E) = E$ and $Ker(P_E) = E^\bot$.
In this way, for each closed subspace $E$ of ${\cal H}$ we denote by $P_E$ the projector onto $E$. Thus, the concepts of closed subspace and projector are interchangeable and, by the usual abuse of language, the identification
\begin{equation}\label{IDENTPROJECSPACE}
E \cong P_E
\end{equation}
will be used along this article.

We denote by ${\cal P}({\cal H})$ the set of all closed projectors on ${\cal H}$. We also remark that ${\cal P}({\cal H}) \subseteq {\cal B}({\cal H})$.   It is well known that the structure
\begin{equation}\label{HILBERTLATTICE}
\langle {\cal P}({\cal H}), \lor, \land, ^\bot, 0, 1_{\cal H}  \rangle
\end{equation}
where $P_X \lor P_Y = P_{X \lor Y}$ being $X \lor Y$ the smallest closed subspace of ${\cal H}$ containing $X$ and $Y$, $P_X \land P_Y = P_{X \cap Y}$, $0 \cong \{0\}$ and $1_{\cal H} \cong {\cal H}$, is a complete orthomodular lattice. These kind of structures are called {\it Hilbert lattices}.

\begin{prop}\label{SUMPROJ}{\rm \cite[\S 28]{HAL}}
Let ${\cal H}$ be a Hilbert space and $P_X$, $P_Y$ two projectors. Then, $P_X + P_Y$ is a projector iff $X\bot Y$. If this condition is satisfied then $P_{X\lor Y} = P_X + P_Y$

\qed
\end{prop}

In \cite{SASAKI} Sasaki noticed that for two closed subspaces $E,X$ of ${\cal H}$ the orthogonal projection of $X$ to the subspace $E$ i.e., $P_E(X)$, can be expressed without the use of the inner
product, more precisely, by using only the lattice operations and orthocomplements of a Hilbert lattice. Formally, taking into account the Sasaki projection introduced in Eq.(\ref{SASAKIPROJECTION}), we have that
\begin{equation}\label{SASAKIPROJECTION2}
P_E(X) = \mu_E(X) = E \land (E^\bot \lor X).
\end{equation}
Thus, in a Hilbert lattice the Sasaki projection $\mu_E$ exactly encodes the action of the orthogonal projector $P_E$ that projects onto the subspace $E$.\\

Let $W$ be an operator in  ${\cal B}({\cal H})$. Then $W$ is said to be {\it unitary} iff $W^* W = W W^* = 1_{\cal H}$. The operator $W$ is an {\it isometry} iff $\Vert W(x) \Vert = \Vert x \Vert $ for every $x\in {\cal H}$. Let us notice that unitary operators are isometries.

\begin{prop}\label{WW1}
Let $W \in {\cal B}({\cal H})$ be an isometry and $X$ be a closed subspace of ${\cal H}$. Then $W(X)$ is a closed subspace of ${\cal H}$.

\end{prop}

\begin{proof}
We shall prove that $W(X)$ contains all its accumulation points. Let $\big(W(x_n)\big)_{n\in \mathbb{N}}$ be a sequence in $W(X)$ which is convergent to $y\in {\cal H}$. Since $\big(W(x_n)\big)_{n\in \mathbb{N}}$ is a Cauchy sequence and $W$ preserves distances we have that  $\Vert x_n - x_m \Vert = \Vert W(x_n) - W(x_m) \Vert \rightarrow 0$ whenever $n,m \rightarrow \infty$. Therefore, $(x_n)_{n\in \mathbb{N}}$ is a Cauchy sequence in $X$ and $x_n$ converges to $x\in {\cal H}$ because ${\cal H}$ is complete. Note that $x\in X$ because $X$ is closed. Then, by the continuity of $W$, we have that $W(x) = W(\lim_{n\rightarrow \infty} x_n) = \lim_{n\rightarrow \infty} W(x_n) = y$. Hence $y\in W(X)$ and $W(X)$ is a closed subspace of ${\cal H}$.

\qed
\end{proof}

An operator $W$ in ${\cal B}({\cal H})$ is a {\it partial isometry} iff the restriction $W\upharpoonright_{ Ker^\bot(W)}$ is an isometry.  There are many familiar examples of partial isometries: every isometry is one and every projection is one also. Partial isometries have been intensively studied since the early stages of the theory of operator on Hilbert spaces. In particular, they form the cornerstone of the dimension theory of von Neumann algebras. One of the earliest results in operator theory is the following characterization of partial isometries.

\begin{theo}\label{EQISOM} {\rm \cite[\S 2.2.8]{PENDERSEN}}
Let $W \in {\cal B}({\cal H})$. Then the following statements are equivalent,

\begin{enumerate}
\item  $W$ is a partial isometry,

\item  $WW^* = P_{_{Imag(W)}}$,

\item  $W^*W = P_{_{Ker(W)^\bot}}$,

\item  $WW^*W = W$,

\item  $W^*WW^* = W^*$,

\item  $W^*$  is a partial isometry.

\end{enumerate}

\qed
\end{theo}

Let ${\cal H}$ be a Hilbert space. A {\it $*$-subalgebra} of ${\cal B}({\cal H})$ is a linear subspace ${\cal N}$ of ${\cal B}({\cal H})$ closed by the operations $\cdot$ and $^*$. In particular ${\cal N}$ is said to be {\it unital} iff the identity operator $1_{\cal H} \in {\cal N}$. Let ${\cal N} \subseteq {\cal B}({\cal H})$. The {\it commutant of ${\cal N}$} is the set $${\cal N}' = \{X\in {\cal B}({\cal H}): \forall T \in {\cal N}, \hspace{0.2cm}  X T = T X \}.$$ In particular ${\cal N}'' = ({\cal N}')'$ is called the {\it bicommutant of ${\cal N}$}.

\begin{definition}
{\rm
Let ${\cal H}$ be a Hilbert space. A {\it von Neumann algebra} is a unital {\it $*$-subalgebra} ${\cal N}$ of ${\cal B}({\cal H})$ such that ${\cal N}'' = {\cal N}$.
}
\end{definition}

Let ${\cal N} \subseteq {\cal B}({\cal H})$ be a von Neumann algebra. We denote by ${\cal P}({\cal N})$ the set of all projectors in ${\cal N}$ i.e. ${\cal P}({\cal N}) = {\cal N} \cap {\cal P}({\cal H})$. An important fact is that ${\cal P}({\cal N})$ generates ${\cal N}$ in the sense that ${\cal P}({\cal N})'' = {\cal N}$. If $A \in {\cal N}$ then we define the {\it left annihilator} of $A$ as the set
\begin{equation}\label{ANN}
\mathit{Ann_L(A)} = \{S\in {\cal N}:SA =0 \}.
\end{equation}

\begin{prop}\label{RIGHTIDEAL}
Let ${\cal N} \subseteq {\cal B}({\cal H})$ be a von Neumann algebra and $A \in {\cal N}$. Then there exists a unique $A' \in {\cal P}({\cal N})$ such that $$\mathit{Ann_L(A)} =  {\cal N}A'  = \{NA' : N\in {\cal N} \}. $$

\end{prop}

\begin{proof}
See {\rm \cite[Definition 37.3 and Remark 37.15]{MM}}.
\qed
\end{proof}

\begin{prop}\label{VONNEUMANNLAT} {\rm \cite[\S D, p.72]{HOLAND}}
Let ${\cal N} \subseteq {\cal B}({\cal H})$ be a von Neumann algebra. Then, ${\cal P}({\cal N})$ is complete orthomodular lattice with respect to the lattice operations inherited from ${\cal P}({\cal H})$.
\qed
\end{prop}

If  ${\cal N} \subseteq {\cal B}({\cal H})$ is a von Neumann algebra then the  orthomodular structure related to ${\cal P}({\cal N})$ is called {\it von Neumann lattice}.
By the above proposition we can also identify the projectors of ${\cal P}({\cal N})$ with the respective closed subspaces of ${\cal H}$.

A crucial relation between projectors in operator theory is the notion of unitary equivalence. Two projectors $P$ and $Q$ in a von Neumann algebra ${\cal N} \subseteq {\cal B}({\cal H})$ are said to be {\it unitary equivalent} iff there exists an unitary operator $W\in {\cal N}$ such that $Q = WPW^*$.

\begin{theo}\label{EQPERSPECUNIT}  {\rm \cite[Theorem 1]{FILLMORE}}.
Let ${\cal N} \subseteq {\cal B}({\cal H})$ be a von Neumann algebra and $P, Q \in {\cal P}({\cal N})$. Then $P$ and $Q$ are unitary equivalent iff $P\sim_p Q$ in the lattice ${\cal P}({\cal N})$.

\qed
\end{theo}

The above theorem allows us to see the notion of perspectivity as a lattice order representation for the notion of unitary equivalence between projectors of a von Neumann algebra. \\

Let ${\cal N} \subseteq {\cal B}({\cal H})$ be a von Neumann algebra. The {\it center of ${\cal N}$}  is defined as the set ${\cal Z}({\cal N}) = {\cal N} \cap {\cal N}'$. Note that ${\cal Z}({\cal N})$ is a commutative  von Neumann  sub algebra of ${\cal N}$.  The algebra ${\cal N}$ is called {\it factor} iff its center is trivial, that is, ${\cal Z}({\cal N}) = \{\lambda I:   \lambda \in {\mathbb{C}} \}$. An example of a factor is ${\cal B}({\cal H})$ for any separable Hilbert space ${\cal H}$.

The notion of factor is closely related to the directly decomposability of the lattice ${\cal P}({\cal N})$. Indeed: Let us notice that each von Neumann algebra ${\cal N}$ has an underlying ring structure, more precisely, it is a  Baer$^*$-ring {\rm \cite[Remark 37.15]{MM}}. Then, ${\cal N}$ is a factor iff the central idempotent elements of the underlying ring structure of ${\cal N}$ are $\{0, 1_{\cal H}\}$ which is equivalent to $Z({\cal P}({\cal N})) = {\bf 2}$. Thus, by Proposition \ref{DIOML}-5, we can establish the following result.

\begin{prop}\label{CENTERFACTOR}
Let ${\cal N} \subseteq {\cal B}({\cal H})$ be a von Neumann algebra. Then the following statements are equivalent:

\begin{enumerate}
\item
${\cal N}$ is a factor.

\item
${\cal P}({\cal N})$ is a directly indecomposable lattice i.e. $Z({\cal P}({\cal N})) = {\bf 2}$.

\end{enumerate}
\qed
\end{prop}

Let ${\cal N} \subseteq {\cal B}({\cal H})$ be a von Neumann algebra. Two projectors $P,Q \in {\cal P}({\cal N})$ are said to be  {\it Murray-von Neumann equivalent} iff there exists a partial isometry $W \in {\cal N}$ such that $WW^* = P$ and  $W^*W = Q$. If  $P$ and $Q$ are  Murray-von Neumann equivalent then we write $P \sim Q$.

\begin{prop}\label{PERSPMURRAY}
Let ${\cal N} \subseteq {\cal B}({\cal H})$ be a von Neumann algebra and $P,Q \in {\cal P}({\cal N})$. If $P\sim_p Q$ then $P\sim Q$.

\end{prop}

\begin{proof}
Let us suppose that $P\sim_p Q$. By Proposition \ref{EQPERSPECUNIT} there exists a unitary operator $U \in {\cal N} $ such that $Q = UPU^*$.
If we define $W = UP$ then $WW^*W = UP(UP)^*UP = UPP^*U^*UP = UPP(U^*U)P = UPP1_{{\cal H}}P = UP =W$. Thus, by Proposition \ref{EQISOM}-4, $W$ is a partial isometry. We also note that
$W^* = (UP)^* = P U^* = (U^*QU)U^* = U^*Q$. Hence $WW^* = Q$, $W^*W = P$ and $P\sim Q$.

\qed
\end{proof}

The relation $\sim$ defines an equivalence on ${\cal P}({\cal N})$  and the description of the quotient ${\cal P}({\cal N}) /_{\sim}$ is known as the {\it dimension theory} for ${\cal N}$.  We write $P \preceq Q$, and say that $P$ is {\it Murray-von Neumann sub-equivalent} to $Q$ if there exists a partial isometry $W \in {\cal N}$ such that $WW^* = P$ and  $W^*W \leq Q$. The projector $P$ is said to be {\it finite} iff, whenever  $P \sim Q$ and $Q \preceq P$ then $P=Q$. Otherwise $P$ is said to be {\it infinite}.

If ${\cal N} = {\cal B}({\cal H})$ then it is immediate to see that two projections are equivalent iff their image have the same dimension. Thus, the ordering in ${\cal P}({\cal H})/_{\sim}$ is isomorphic to $\{0,1,\ldots,n \}$ if $dim({\cal H}) = n$ and isomorphic to the ordinal number ${\mathbb{N}} \cup \{\infty\}$ if ${\cal H}$ is a infinite separable dimensional Hilbert space. Hence, the idea of equivalence of projectors via partial isometries represents an abstract notion of dimension for an arbitrary factor an the first result confirming this is the following theorem whose proof can be found in \cite{MV}.

\begin{prop} \label{TOTPRECEC}
Let  ${\cal N} \subseteq {\cal B}({\cal H})$ be a factor. Then $\langle {\cal P}({\cal N}) /_{\sim}, \preceq \rangle$ is a totally ordered set.
\qed
\end{prop}

Let us notice that the order $\preceq$ is defined through the ring structure of a factor. Thus, two factors can not be isomorphic if the orderings of the corresponding quotient $/_{\sim}$ are different.

\begin{prop}\label{FILLLEMMA}{\rm \cite[Lemma 4]{FILLMORE}}
Let ${\cal N} \subseteq {\cal B}({\cal H})$ be a factor and $P,Q, R \in {\cal P}({\cal N})$ such that $Q\sim P \leq Q$ and $P \preceq R \leq Q^\bot$. Then $P$ and $Q$ are perspective in $[0, Q\lor R]_{{\cal P}({\cal N})}$.
\qed
\end{prop}

The Murray-von Neumann classification of factors consists in determining the order type of ${\cal P}({\cal N}) /_{\sim}$ for a factor ${\cal N}$.
In order to do this, the following notion of dimension function introduced in {\rm \cite[Definition 8.2.1]{MV}}  plays a crucial role.

\begin{definition}\label{DIMFUNCTION}
{\rm Let  ${\cal N} \subseteq {\cal B}({\cal H})$ be a von Neumann algebra. A function of the form $D: {\cal P}({\cal N}) \rightarrow [0, \infty]$ is called {\it dimension function} if it satisfies the following
requirements:

\begin{enumerate}
\item
$D(P) = 0$ iff $P=0$,

\item
$P \sim Q$ iff $D(P) = D(Q)$,

\item
if $P \bot Q$ then $D(P + Q) = D(P) + D(Q)$.

\end{enumerate}
}
\end{definition}

\begin{theo}\label{TYPECLASSIFICATION} {\rm \cite[Theorem VII]{MV}}
Let ${\cal N} \subseteq {\cal B}({\cal H})$ be a factor. Then there exists a dimension function $D: {\cal P}({\cal N}) \rightarrow [0, \infty]$ uniquely determined up to positive constant multiple. Further, $Imag(D)$ falls into exactly one of four possible cases, depending on which of the following sets is the range of some scaling of $D$.

\begin{itemize}

\item[] Type ${\bf I}_n${\rm :} $Imag(D) = \{0,1 \ldots n\}$.

\item[] Type ${\bf I}_\infty${\rm :} $Imag(D) = \{0,1 \ldots _\infty \}$.

\item[] Type ${\bf II}_1${\rm :} $Imag(D) = [0,1]$.

\item[] Type ${\bf II}_\infty${\rm :} $Imag(D) = [0,\infty]$.

\item[] Type ${\bf III}${\rm :} $Imag(D) = \{0,\infty \}$.
\end{itemize}

\qed
\end{theo}

The above theorem implies that if  ${\cal N} \subseteq {\cal B}({\cal H})$ is a factor then $\langle {\cal P}({\cal N}) /_{\sim}, \preceq \rangle$ is order isomorphic to the image of the uniquely determined dimension function $D: {\cal P}({\cal N}) \rightarrow [0, \infty]$.  In this way all von Neumann factor were found to belong to the classes type ${\bf I}$ or type ${\bf II}$ or type ${\bf III}$. The following proposition, sometime referred as Borchers condition \cite{BRO, DOP}, provides a crucial characterization of type III factors.

\begin{prop}\label{RELATED1}
Let ${\cal N} \subseteq {\cal B}({\cal H})$ be a factor such that ${\cal N} \not = {\cal Z}({\cal N})$ and $D: {\cal P}({\cal N}) \rightarrow [0, \infty]$ be the dimension function. Then the following statement are equivalent,

\begin{enumerate}
\item
${\cal N}$ is a type III factor.

\item
For all $P_1, P_2 \in {\cal P}({\cal N}) - \{0\}$, $P_1 \sim P_2$.

\item
For each $P \in  {\cal P}({\cal N}) -\{0\}$, $P \sim 1_{\cal H}$. \hspace*{\fill} {{\rm ({\it Borchers condition})} }

\end{enumerate}
\qed
\end{prop}

\begin{prop} \label{INFINITEDIM}
Let ${\cal N} \subseteq {\cal B}({\cal H})$ be a factor and $P_X,P_Y \in {\cal P}({\cal N}) - \{0, 1_{\cal H}\}$ such that $ X \subseteq Y$. Then

\begin{enumerate}

\item
$D(P_Y) = D(P_X) + D(P_{X^\bot \cap Y})$.

\item
If $P_X \sim_p P_Y$ and $X \not= Y$ then $D(P_X) = D(P_Y) = D(1_{\cal H}) = \infty$.

\end{enumerate}

\end{prop}

\begin{proof}
1) By the orthomodular law $Y = X \lor (X^\bot \cap Y)$ and $X \bot (X^\bot \cap Y)$. Thus, by Proposition \ref{SUMPROJ}, $P_Y = P_{X \lor (X^\bot \cap Y)} = P_X + P_{X^\bot \land Y}$. Hence, by condition 3 in Definition \ref{DIMFUNCTION}, we have that $D(P_Y) = D(P_X) + D(P_{X^\bot \cap Y})$.

2) By Proposition \ref{PERSPMURRAY} we have that $D(P_X) = D(P_Y)$ because $X \sim_p Y$. Let us suppose that $D(P_Y) < \infty$.  Thus, by item 1, we have that $$D(P_Y) = D(P_X) + D(P_{X^\bot \land Y}) = D(P_Y) + D(P_{X^\bot \land Y})$$ and consequently, $D(P_{X^\bot \land Y}) = 0$. Then, by Condition 1 in Definition \ref{DIMFUNCTION}, $X^\bot \cap Y = X^\bot \land Y = \{0\}$. Let us notice that $X^\bot \cap Y = \neg_Y X$ is the complement of $X$ in the interval lattice $[0,Y]_{{\cal P}({\cal N})} \approx [0,P_Y]_{{\cal P}({\cal N})}$. Since $X \lor \neg_Y X = Y$ and $X < Y$ then $X^\bot \cap Y = \neg_Y X \not= 0$ which is a contradiction. Hence, $D(P_Y) = D(P_X) = \infty$ and, by item 1,
$D(1_{\cal H}) = D(P_Y + P_{Y^\bot}) = D(P_Y) + D(P_{Y^\bot}) = \infty + D(P_{Y^\bot}) = \infty$.

\qed
\end{proof}

\section{Motivating $LQF$-logic:  Logics from Murray-von Neumann dimension theory and related open questions}\label{RINGOP}

In order to motivate our logical system we first need to introduce a brief description of the AQFT whose formalism is supported by two theories: the algebraic approach to quantum mechanics and general relativity.

On one side, the standard operator algebra formulation of quantum mechanics, or {\it algebraic quantum mechanics} for short, starts with the concept of separable complex Hilbert space. This formalism was described in a kind of postulates by J. von Neumann in his celebrated 1932 book \cite{vNlibro}. In this approach we can understand quantum mechanics as a theory whose primitive concepts are {\it quantum system}, {\it states}, {\it observables}  and {\it measurement} subject to the following basic interpretations:

\begin{itemize}
\item
To every quantum system, we associate a complex separable Hilbert space ${\cal H}$ called the {\it state space of the system}.
\end{itemize}

Intuitively speaking,  a physical system  consists of a region of  spacetime and all the entities   contained within it.

\begin{itemize}
\item
The {\it state of the system} is completely specified by a density operator $\rho$ on ${\cal H}$, that is, a self-adjoint ($\rho = \rho^*$), positive ($\langle -, \rho(-) \rangle \geq 0$) and unit-trace ($Tr(\rho) = 1$) operator.
\end{itemize}

The state of a quantum system encodes all of the physical information that we know about the system. Traditional quantum mechanics distinguishes between pure states and mixed states. If our knowledge of the state of the system is complete we say that the system is in a {\it pure state}. This special case is represented by a density operator degenerates to a projector on a closed $1$-dimensional subspace of ${\cal H}$. Thus, a pure state is equally well characterized by a unit vector which defines this $1$-dimensional subspace. Otherwise, if our knowledge of the state is incomplete but the statistical ensemble associated to the system is known, we say that the system is in a {\it mixed state}.

\begin{itemize}
\item
The {\it observables} of a quantum system are represented by self-adjoint operators on the space ${\cal H}$.
\end{itemize}

Intuitively, an observable represents a physical property ${\cal A}$ as energy, position, momentum, etc. If ${\cal A}$ is represented by  the self-adjoint operator $A$ then the eigenvalues of $A$ are real numbers representing all the possible values of the physical property ${\cal A}$. \\

There are other postulates, involving measurement process and time evolution of a quantum system, that we do not mention because they do not play any role in the logical systems  treated here.  \\

Let us notice that a projector is a self-adjoint operator having only the two eigenvalues $0$ and $1$. Along these lines, von Neumann proposed a correspondence between projectors and logical propositions in his 1932 book \cite{vNlibro}. We refer to such propositions as {\it quantum propositions} or {\it q-propositions} for short. In order to clarify this concept, let us consider the spectral decomposition of a self-adjoint operator $A$ representing an observable ${\cal A}$ i.e.
\begin{equation}\label{SPECTRALDEC}
A= \sum_i a_i P_{a_i}, \hspace{1cm} \sum_i P_{a_i} = 1_{\cal H}
\end{equation}
where $(a_i)_i$ are the eigenvalues of $A$ and $(P_{a_i})_i$ is the family of projectors over the respective eigenspaces. In this way, each projector $P_{a_i}$ represents the proposition ``{\it the value of the observable ${\cal A}$ is $a_i$}". Thus, by the spectral decomposition, each observable can be decomposed into elementary true/false propositions.
We can also analyze under which conditions these kind of propositions are true or false. Indeed: Let us consider  an arbitrary state of the system represented by a density operator $\rho$. In this case we only know that, in the state $\rho$, the measurement of ${\cal A}$ will yield one of the values $a_i$, but we don't know which one. Then, for each projector $P_{a_i}$ of the spectral decomposition, mentioned in Eq.(\ref{SPECTRALDEC}), the number in the real interval $[0,1]$ given by the trace $Tr(\rho P_{a_i})$ represents the probability that the proposition ``{\it the value of the observable ${\cal A}$ is $a_i$}" is true in the state $\rho$.
In the particular case where $Tr(\rho P_{a_i}) = 1$, for example when $\rho = P_{a_i}$, we can regard the proposition ``{\it the value of the observable ${\cal A}$ is $a_i$}" as true in the state $\rho$. But, if $Tr(\rho P_{a_i}) = 0$, as is the case of $\rho = P_{a_i}^\bot$, a measurement of ${\cal A}$ will never provide result $a_i$. In this case we can consider the proposition ``{\it the value of the observable ${\cal A}$ is $a_i$}" as a false proposition in the state $\rho$. Therefore, projectors correspond to true/false propositions but they are not organized
in a Boolean structure. For example, the distributivity condition between projectors fails. In this way, the successive work in logic related to quantum systems, developed by Birkhoff and von Neumann in their seminal article in 1936 \cite{BvN}, substituted Boolean algebras with the lattice of closed subspaces of a Hilbert space, or Hilbert lattice, for encoding the structure of q-propositions. This structure was successively named {\it quantum logic}.

Soon after the publication of von Neumann book in 1932, his interest in ergodic theory, group representations and quantum mechanics contributed significantly to von Neumann realization that a theory of operator algebras was the next important stage in the development of quantum mechanics \cite{REDEI}. Operator algebras can trace its origin to the appearing of the four ``{\it Rings of Operators}" papers of F. J. Murray and J. von Neumann, where the first one \cite{MV} appeared in 1936. Those operator algebras, originally known as rings of operators or  $W^*$-algebras, were later  renamed von Neumann algebras by J. Dixmier and J. Dieudonn\'e. That was when the orthomodular lattice theory grew out of the theory of von Neumann algebras as a general framework to include all projection lattices of these algebras. More precisely, the Murray-von Neumann classification of factors suggests that the projector lattices of each factor can be quite different from one on another giving rise to a wide family of lattices whose common characteristic is the orthomodularity. Indeed:

Type $I_n$ and type $I_\infty$  factors always correspond to the whole algebra of bounded operators on a separable Hilbert space.  This is the case usually considered for describing quantum systems with a state space having finite or infinite dimension, respectively. The standard model of type $I$ factor is the algebra ${\cal B}({\cal H})$  where ${\cal P}({\cal B}({\cal H})) = {\cal P}({\cal H})$. Thus, the logic associated to this factor is the usual Birkhoff and von Neumann quantum logic.  If $Dim({\cal H}) = n$ then ${\cal P}({\cal B}({\cal H}))$ is an atomic modular orthomodular lattice. Differently, if $Dim({\cal H}) = \infty$ then ${\cal P}({\cal B}({\cal H}))$ loses the modularity  {\rm \cite[\S 10.3.8]{BELT}}.

Type $II_1$ and type $II_\infty$ factors play an important role in non-relativistic statistical quantum mechanics. The projector lattice of a $II_1$ factor is a modular orthomodular lattice  and it has no atoms. While, the projector lattice of a type $II_\infty$ factor is non modular orthomodular lattice and it has no atoms too.

Type $III$ factors was the most mysterious case, however,  they became relevant in relativistic quantum field theories.  At the beginning, von Neumann regarded the type $III$ factors as a kind of pathological class of operator algebras. Indeed, it took four years after the discovery of the classification of factors, in 1936,  to construct the first example of type III factor \cite{vonNeumannIII}. Moreover, it was only in the mid of the sixties that the existence of a continuous number of non isomorphic type $III$ factors was proven \cite{POWER}. The projector lattice of a type $III$ factor is a non modular orthomodular lattice and it has no atoms. In the next section Proposition \ref{RELATED2} provides an alternative proof of these lattice order properties about type $III$ factors.

The other theory which AQFT is based on is the  theory of relativity. In general relativity, the space of events is represented by a {\it Lorentzian manifold}, i.e. a pair $\langle {\cal M}, g \rangle$ where ${\cal M}$ is a smooth $(n + 1)$-dimensional manifold and $g$ is a {\it Lorentzian tensor metric}, that is, $g$ associates to each point $p\in {\cal M}$ the following inner product $g_p(-,-)$ on the tangent space $T_p{\cal M}$: chosen a basis $\{e_0 \ldots e_n  \}$ in $T_p{\cal M}$ and $x,y \in T_p{\cal M}$, $g_p(x,y) = -x^0y^0 + \sum^n_{i=1}x^iy^i $ where $(x^0\ldots x^n)$ and $(y^0\ldots y^n)$ are the components of  $x$, $y$ with respect to the chosen basis. For each $x\in T_p{\cal M}$ the tensor metric defines the following classification:
\begin{equation}
x \hspace{0.2cm} {\rm is} \hspace{0.2cm} \begin{cases}{\it spacelike} &\text{iff} \hspace{0.2cm} g_p(x,x) >0 \\
{\it null} & \text{iff}\hspace{0.2cm} g_p(x,x)= 0\\
{\it timelike} &\text{iff} \hspace{0.2cm} g_p(x,x) <0. \end{cases}
\end{equation}
The {\it causal cone at $p \in {\cal M}$} is defined as $C_p = \{x\in T_p{\cal M}: g_p(x,x) \leq 0 \}$. A  {\it causal curve} is a smooth curve $\gamma: I \rightarrow {\cal M}$ such that for each $s\in I$, $\dot{\gamma}(s) \in C_{\gamma(s)}$. An important principle of general relativity states that observers can move only on causal curves.  Let ${\cal O}$ be an open bounded region in ${\cal M}$ and $x\in {\cal M}$. We said that $x$ lies in the causal cone of ${\cal O}$, denoted by $C({\cal O})$, iff there exists $p\in {\cal O}$ and a causal curve $\gamma:[0,1] \rightarrow {\cal M}$ such that $\gamma(0) = p$ and $\gamma(1) = x$.  The {\it causal complement} of the open bounded region ${\cal O}$ is defined as ${\cal O}^c = {\cal M} \smallsetminus \overline{C({\cal O})}$. We say that an open bounded region ${\cal O}_1$ is {\it spacelike separated from} the open bounded region ${\cal O}_2$ iff ${\cal O}_1 \subseteq {\cal O}_2^c$. In general relativity events happening at spacelike separated regions cannot influence each other. A {\it space time} or {\it timelike oriented} Lorentzian manifold is a $3$-tuple $\langle {\cal M}, g, {\bf t} \rangle$ such that $\langle {\cal M}, g \rangle$ is a Lorentzian manifold and ${\bf t}$ is a smooth vectorial field on ${\cal M}$ i.e.,  $p \mapsto {\bf t}(p) \in T_p{\cal M}$, such that $g_p({\bf t}(p),{\bf t}(p)) < 0$. The vector field ${\bf t}$ determines a time orientation on ${\cal M}$. More precisely, if $\gamma:I \rightarrow M$ is a causal curve then we say that $\gamma$ is  {\it future directed} (resp. {\it past directed}) provided for each $s\in I$, $g_{\gamma(s)}({\bf t}(\gamma(s)),\dot{\gamma}(s) ) >0$ (resp. $g_{\gamma(s)}({\bf t}(\gamma(s)),\dot{\gamma}(s) ) <0$). \\

These geometrical notions on relativity and the basic concepts of the algebraic formulation of quantum mechanics introduced above allow us to briefly describe the Haag-Araki formalism of AQFT and a related logical systems.

Let ${\cal M}$ be a spacetime represented by a Lorentzian manifold, $\mathfrak{R}({\cal M})$ be the set of open bounded region and ${\cal H}$ be a separable Hilbert space. The basic object of the Haag-Araki model for AQFT is a net of von Neumann algebras, called {\it net of observables over the spacetime ${\cal M}$} defined as  $$\mathfrak{R}({\cal M})\ni {\cal O} \mapsto {\cal N}({\cal O}) $$ where ${\cal N}({\cal O}) \subseteq {\cal B}({\cal H}) $ is a Type $III$ factor and the following basic conditions are satisfied
\begin{enumerate}
\item[]  If ${\cal O}_1 \subseteq {\cal O}_2$ then ${\cal N}({\cal O}_1) \subseteq {\cal N}({\cal O}_2)$. \hspace*{\fill} {({\it Isotony})}

\item[]  If ${\cal O}_1 \subseteq {\cal O}_2^c $ then ${\cal N}({\cal O}_1) \subseteq {\cal N}({\cal O}_2)'$. \hspace*{\fill} {({\it Locality})}

\end{enumerate}

Let us remark that there are other conditions imposed on the net of observables \cite{HAG, HALV, HORRU} that we do not mention because they do not play any role in the argument treated here.

Each factor ${\cal N}({\cal O})$ is called {\it local observable algebra of the net} and it represents the observables in the region ${\cal O}$ of the spacetime. In this perspective, projectors in ${\cal P}\big( {\cal N}({\cal O})\big)$ represent q-propositions related to ${\cal O}$. A {\it state of the net} is a density operator of the von Neumann algebra generated by $\bigcup_{{\cal O} \in \mathfrak{R}({\cal M}) } {\cal N}({\cal O}) \subseteq {\cal B}({\cal H}) $.

The role of type III factors in the net of observables turns out to be a consequence of general physical requirements compatible with the relativity. For example, one of these is the fact that measurements of observables happening in space-like separated regions are not allowed to influence each other. The next proposition formally expresses the mentioned physical requirement.

\begin{prop}\label{SEPARATEDPORP}
Let $\mathfrak{R}({\cal M})\ni {\cal O} \mapsto {\cal N}({\cal O})$ be a net of observables in a Hilbert space ${\cal H}$, $\rho$ be a state of the net and $P$ be a projector of the local algebra ${\cal N}({\cal O})$. Then there exists a state $\rho_{_P}$ of the net such that

\begin{enumerate}
\item
$tr(\rho_{_P} P) = 1$.

\item
If ${\cal O} \subseteq {\cal O}_1^c$ and $A$ is a selfadjoint operator in ${\cal N}({\cal O}_1)$ then $tr(\rho_{_P} A) = tr(\rho A)$.

\end{enumerate}

\end{prop}

\begin{proof}
Let $\rho$ be a state of the net and $P$  be a  projector in the local algebra ${\cal N}({\cal O})$. Since ${\cal N}({\cal O})$ is a type III factor then, by Proposition \ref{RELATED1}-3, there exists a partial isometry $W \in  {\cal N}({\cal O})$ such that $W W^* = 1_{\cal H}$ and $W^* W  = P$. Let us define $\rho_{_P} = W^* \rho W$. We first show that $\rho_{_P}$ is a density operator. Indeed: By a straightforward calculations we can see that  $\rho_{_p}^* = \rho_{_p}$.  In order to prove that $\rho_p$ is positive, let us consider $x\in {\cal H}$ and, by Proposition \ref{closed1}, the orthogonal decomposition $x = x_{_{W}} + x_{_{W^\bot}}$ where $x_{_{W}} \in Ker(W)$ and $x_{_{W^\bot}} \in Ker(W)^\bot$.
Since  $W^* \rho W(x_{_{W^\bot}}) \in ker(W)^\bot $ and $W$ is an isometry over $ker(W)^\bot$ then we have that  $\langle x, \rho_{_P}(x) \rangle = \langle x, W^* \rho W(x) \rangle = \langle x_{_{W^\bot}}, W^* \rho W(x_{_{W^\bot}}) \rangle + 0 = \langle W(x_{_{W^\bot}}), W W^* \rho W(x_{_{W^\bot}}) \rangle = \langle W(x_{_{W^\bot}}), \rho W(x_{_{W^\bot}}) \rangle$. Thus we also have that $\langle x, \rho_{_P}(x) \rangle = \langle W(x_{_{W^\bot}}), \rho W(x_{_{W^\bot}}) \rangle \geq 0 $ since $\rho$ is a positive operator. Note that $\rho_{_P}$ is a unit-trace operator because $tr(\rho_{_P}) = tr(W^* \rho W) = tr(WW^* \rho ) = tr(\rho) = 1$. Hence, $\rho_{_p}$ is a density operator.

Now we prove that $tr(\rho_{_P} P) = 1$. Indeed:
\begin{eqnarray*}
tr(\rho_{_P} P) & = & tr(W^* \rho W P) =  tr(P W^* \rho W) = tr(W^* W W^* \rho W) = \\
& = & tr(W^* \rho W) = tr(\rho_{_P}) = 1.
\end{eqnarray*}
Let us suppose that  ${\cal O} \subseteq {\cal O}_1^c$ and $A$ is a selfadjoint operator in ${\cal N}({\cal O}_1)$. Note that $WA = A W$ since, by the postulate of locality in the net of observables, $A$ belongs to the commutant of ${\cal N}({\cal O})$. Thus, $Tr(\rho_{_P} A) = Tr(W^* \rho W A) = Tr(W^* \rho A W ) = Tr(WW^* \rho A) = Tr(\rho A)$. Hence our claim.

\qed
\end{proof}

As we have seen in the above proof, the type III factor condition imposed to each algebra ${\cal N}({\cal O})$ implies that every projector $P \in {\cal N}({\cal O})$ can be written as $W^* W  = P$ where  $W \in {\cal N}({\cal O})$ is a partial isometry. Consequently, if $\rho$ is a state of the net then the transformation  $\rho \mapsto \rho_{_P} = W^* \rho W$ change the state $\rho$ into the eigenstate $\rho_{_P}$ of $P$ by the local operation $W$ without disturbing the causal complement of the region we are dealing with. This is precisely the motivation of the type III factor in order to show that the relativity is not violated. For more details regarding the role of type III factors in quantum field theory we refer to \cite{YNG}.

Proposition \ref{SEPARATEDPORP} also provides an interesting insight into the local character of the logic of q-propositions related to each region ${\cal O}$ of the Lorentzian spacetime. Indeed: Let us consider a q-proposition, represented by the projector $P \in {\cal N}({\cal O})$. Since every state $\rho$ of the net can be changed into an eigenstate $\rho_{_P}$ of $P$ via a local operation, the q-proposition $P$ can be evaluated as true in the state $\rho_{_P}$ i.e., $tr(\rho_{_P} P) = 1$ without affecting the truth values of the q-propositions of the spacelike separated regions of ${\cal O}$. In other words, the type III condition of each algebra of the net defines a local logical system such that in spacelike separated regions of the spacetime the respective notions of truth are independent each other. In this way, the goal of the $LQF$-logic is to capture the type III factor condition of the local observable algebras.

In the aim to develop a such logical system, unavoidably, we deal with a family of open questions related to the problem of characterizing those logical systems which can be identified with the lattice of projections of factors of the Murray-von Neumann classification. Indeed: Since the beginning of the Murray-von Neumann classification, in the 30's, it has been noticed that each factor of the classification is associated to a particular kind of von Neumann lattice. The subject aroused great interest among the logical community giving rise to the study of logical structures emerging from the factors of the classification.
At first, von Neumann introduced the concept of continuous geometry in the mid of the thirties as a lattice theoretical generalization of the projective geometry. More precisely, continuous geometries define a subclass of modular lattices. These basic results were published in \cite{vonNeumannGEO} after his death from his Princeton lecture notes during 1935-1937. In this work it was shown that in a directly indecomposable continuous geometry the quotient by the perspectivity relation $\sim_p$ univocally defines a notion of dimension function onto the real interval $[0,1]$.

\begin{rema}\label{REMAMOD}
{\rm  It is interesting to note that in a factor ${\cal N} \subseteq {\cal B}({\cal H})$ the relations $\sim$ and $\sim_p$ coincide whenever $1_{\cal H}$ is a finite element \cite{FILLMORE, KAPL}. Moreover, the condition of the finiteness of $1_{\cal H}$ it can be equivalently formulated in equational terms through the equation of modularity (see Eq.(\ref{MODULARITYEQ})). Thus, the modularity imposed to the lattice of projectors of the factor ${\cal N}$ forcing it to be either a type $I_n$ factor or a type $II_1$ factor (see {\rm \cite[\S 10.3.8]{BELT}} for more detail). In this way, the notion of continuous geometries provides a natural common framework for type $I_n$ and $II_1$ factors.
}
\end{rema}

From the remark above, naturally arises the question whether it might be possible to establish lattice theoretical conditions in order to characterize each factor of the Murray-von Neumann classification  \cite{KAPL2, MAEDA}.
In this way and in separated works Loomis \cite{LOOMIS} and Maeda \cite{MAEDA2}, assuming the existence of a particular equivalence relation on an orthomodular lattice, derived several properties related to the Murray-von Neumann dimension theory in a purely algebraic way. This kind of structure is known as {\it dimension lattices}. Although the dimension lattices theory provides a general theoretical framework for von Neumann lattices, from a logical algebraic point of view, to establish an equational theory able to capture the Murray-von Neumann dimension equivalence remains an undone tasks of orthomodular lattice theory. The great difficult in order to deal with this problem is that, in general, we can not transplant the Murray-von Neumann equivalence to a lattice since it involves elements of the theory of operators on Hilbert spaces. However, as we will see in the next section, this is possible to do for the special case of type III factors. Therefore, and in a similar way as is mentioned in Remark \ref{REMAMOD}, $LQF$-logic will be based on an equational system that, when imposed on the projector lattice of a factor it determines univocally the type III factor condition. This result will be proved in Theorem \ref{CHARACTTYPEIII}. In this way, the logical system studied here, that  describes the propositional structure of the local observable algebras in the Haag-Araki model of AQFT, is closely related to the question about whether it might be possible to establish lattice theoretical conditions in order to characterize the factors of the Murray-von Neumann classification.

\section{An equational characterization for the type $III$ factor} \label{EQTYPEIII}

In this section we establish a set of equations able to capture, in a purely algebraic way, the dimension function and the Murray-von Neumann equivalence of a type III factor.
More precisely, the mentioned equational system, imposed to a von Neumann lattice,  univocally determines the type $III$ factor. In order to do this, we first need to establish a representation of the dimension function into the projector lattice of a type $III$ factor.

\begin{definition}\label{INTERNALDIMENTION0}
{\rm Let ${\cal N} \subseteq {\cal B}({\cal H})$ be a type $III$ factor. Then we define the {\it internal dimension function} as the unary operation $w_0:{\cal P}({\cal N}) \rightarrow \{0,1_{{\cal H}}\}$ such that
$$
w_0(X) = \begin{cases}0, & X = 0, \\
1_{{\cal H}}, & otherwise. \end{cases}
$$
}
\end{definition}

The following proposition  can be easily proved.

\begin{prop}\label{INTERNALDIMENTION}
Let ${\cal N} \subseteq {\cal B}({\cal H})$ be a type $III$ factor, $D:{\cal P}({\cal N}) \rightarrow \{0,\infty\}$ be the dimension function and $w_0$ be the internal dimension function. Then
\begin{enumerate}
\item
$w_0(X) = 0$ iff $D(X) = 0$,

\item
$X\sim Y$ iff $w_0(X) = w_0(Y)$,

\item
$w_0(X\lor Y) = w_0(X) \lor w_0(Y)$,

\item
$w_0(w_0(X)) = w_0(X) $,

\item
$w_0(X\land w_0(Y)) = w_0(X) \land w_0(Y)$,

\end{enumerate}
\qed
\end{prop}

\begin{rema}
{\rm Let us notice that items 1 and 2 in the above proposition show that $w_0$ exactly describes the dimension function and the Murray-von Neumann equivalence within the ordered structure of the projectors of a
type $III$ factor. Furthermore, if we consider the category whose objects are the von Neumann lattices of the type $III$ factors expanded by the internal dimension function and whose
arrows are internal dimension function preserving ${\cal OML}$-homomorphisms, then it is equivalent to the category of von Neumann lattices of type $III$ factors whose arrows are the following commutative triangles
\begin{center}
\unitlength=1mm
\begin{picture}(20,20)(0,0)
\put(8,16){\vector(3,0){5}}
\put(3,12){\vector(1,-2){4}}
\put(18,12){\vector(-1,-2){4}}

\put(2,10){\makebox(17,0){$\equiv$}}

\put(0,16){\makebox(0,0){${\cal P}({\cal N}_1)$}}
\put(21,16){\makebox(0,0){${\cal P}({\cal N}_2)$}}
\put(11,0){\makebox(0,0){$\{0,\infty\}$}}
\put(2,20){\makebox(17,0){$f$}}
\put(2,8){\makebox(-6,0){$D_1$}}
\put(24,8){\makebox(-6,0){$D_2$}}
\end{picture}
\end{center}
In other words, the preservation of the internal dimension function is equivalent to the preservation of the Murray-von Neumann equivalence through  ${\cal OML}$-homomorphisms between
von Neumann lattices of type $III$ factors.
}
\end{rema}

\begin{prop}\label{PARCISOM0}
Let $W \in {\cal B}({\cal H})$ be a partial isometry. Then the following  statements are equivalent:

\begin{enumerate}
\item
$WW^* = 1_{{\cal H}}$.

\item
The restriction $ W\upharpoonright_{Ker^\bot(W)}$ defines an isometry onto ${\cal H}$.

\item
$W^*$ defines a surjective isometry of the form $W^*:{\cal H} \rightarrow Ker^\bot(W)$.

\end{enumerate}

\end{prop}

\begin{proof}
$1 \Longrightarrow 2)$ Since $W$ is a partial isometry we only need to prove that  $W\upharpoonright_{Ker^\bot(W)}$ is a surjective map. Let $y \in {\cal H}$. Then, by hypothesis, $y = WW^*(y) = W(x)$ where $x = W^*(y)$. By Proposition \ref{closed1} and Remark \ref{CLOSEDKERIMAG} $x$ can be decomposed as $x = x_{_{Ker(W)}} + x_{_{ker^\bot(W)}}$ where $x_{_{Ker(W)}} \in Ker(W)$ and $x_{_{ker^\bot(W)}} \in Ker^\bot(W)$. Then, $y = W(x) = W(x_{_{Ker(W)}} + x_{_{ker^\bot(W)}}) = W(x_{_{Ker(W)}}) + W(x_{_{ker^\bot(W)}}) = 0 + W(x_{_{ker^\bot(W)}}) = W(x_{_{ker^\bot(W)}})$. Thus, $ W\upharpoonright_{Ker^\bot(W)}$ is a surjective map.

$2 \Longrightarrow 3)$  By hypothesis we have that $Imag(W) = {\cal H}$. Then, by Proposition \ref{KERRANG}-1, $Ker(W^*) = Imag^\bot(W^{**}) = Imag^\bot(W) = {\cal H}^\bot = \{0\}$ i.e.,  $W^*$ is injective. Therefore, $W^*$ is an isometry on ${\cal H}$ and, by Proposition \ref{WW1}, $Imag(W^*)$ is a closed subspace of ${\cal H}$.  Then, by Proposition \ref{KERRANG}-2, $Imag(W^*) = \overline{Imag(W^*)} = Ker^\bot(W)$. In this way $W^*$ defines a surjective isometry of the form $W^*:{\cal H} \rightarrow Ker^\bot(W)$.

$3 \Longrightarrow 1)$
By hypothesis $Ker(W^*) = \{0\}$. Thus, by Proposition \ref{KERRANG}-1, $Imag^\bot(W) = Imag^\bot(W^{**}) = Ker(W^*) = \{0\}$ and then $Imag(W) = {\cal H}$. Hence, by Theorem \ref{EQISOM}-2, $WW^* = P_{_{Imag(W)}} = P_{{\cal H}} = 1_{{\cal H}}$.

\qed
\end{proof}

\begin{prop}\label{PARCISOM1}
Let $W \in {\cal B}({\cal H})$ be a partial isometry such that $WW^* = 1_{{\cal H}}$. Then for each $P_{_X} \in {\cal P}({\cal H})$
\begin{enumerate}

\item
$W(X)$ and $W^*(X)$ are closed subspaces of ${\cal H}$.

\item
$W^*P_{_X}$ is a  partial isometry.

\item
$P_{_{W^*(X)}} = W^*P_{_X}W$.

\end{enumerate}
\end{prop}

\begin{proof}
1) We first prove that $W(X)$ is a closed subset. Since $Ker(W)$ is a closed subspace of ${\cal H}$, by Proposition \ref{closed1}-2, ${\cal H} = Ker(W) \oplus Ker^\bot(W)$. Then each $x \in X$
can be written as $x = x_{_{Ker(W)}} + x_{_{Ker^\bot(W)}}$ where $x_{_{Ker(W)}} \in Ker(W)$ and $x_{_{Ker^\bot(W)}} \in Ker^\bot(W)$. Let us notice that $$X_{_{Ker^\bot(W)}} = \{x_{_{Ker^\bot(W)}} \in {\cal H}: x\in X \} =  P_{_{Ker^\bot(W)}}(X)$$ is a closed subspace of $Ker^\bot(W)$ and $W(X) = W(X_{_{Ker^\bot(W)}})$. Consequently, by Proposition \ref{PARCISOM0}-2, $W(X)$ is a closed subspace of ${\cal H}$. By Proposition \ref{PARCISOM0}-3 and Proposition \ref{WW1} immediately follows that $W^*(X)$ is also a closed subspace of ${\cal H}$.

2) By Proposition \ref{PARCISOM0}-3, $W^*$ is an isometry onto $Ker^\bot(W)$ and therefore $Ker(W^*) = \{0\}$. Thus,  $Ker(W^* P_{_X}) = Ker(P_{_X}) = X^\bot$. Consequently, $ W^* P_{_X}\upharpoonright_{Ker^\bot(W^* P_{_X})} =  W^* P_{_X} \upharpoonright_X = W^*\upharpoonright_X$ is an isometry too. In this way $W^* P_{_X}$ is a partial isometry.

3) By item 2 $W^*P_{_X}$ is a  partial isometry. Then, by Proposition \ref{EQISOM}-2, $P_{_{W^*(X)}} =  P_{Imag(W^*P_{_X})} =  W^*P_{_X}(W^*P_{_X})^* =  W^*P_{_X} P_{_X} W = W^*P_{_X}W$.

\qed
\end{proof}

\begin{prop}\label{PARCISOMCLOSED}
Let ${\cal N} \subseteq  {\cal B}({\cal H})$ be a von Neumann algebra and $W \in {\cal N}$ be a partial isometry such that $WW^* = 1_{{\cal H}}$. Then for each $P_{_X} \in {\cal P}({\cal N})$
$$P_{_{W^*(X)}}\in {\cal N} \hspace{0.2cm} \mbox{and} \hspace{0.2cm} P_{_{W(X)}}\in {\cal N}.$$

\end{prop}

\begin{proof}
By Proposition \ref{PARCISOM1}-3 it is immediate to see that $P_{_{W^*(X)}}\in {\cal N}$. In order to prove that $P_{_{W(X)}}\in {\cal N}$, let us consider the left annihilator of the operator $WP_{_X}$ in the algebra ${\cal N}$, that is, $$Ann_L(WP_{_X}) = \{S\in {\cal N}: S(WP_{_X}) = 0 \}. $$ By Proposition \ref{RIGHTIDEAL} there exists unique projector $(WP_{_X})' \in {\cal P}({\cal N})$ such that $Ann_L(WP_{_X}) = {\cal N}(WP_{_X})'$. We shall prove that $(WP_{_X})' = P_{_{W(X)^\bot}}$.
Indeed:

Let $N P_{_{W(X)^\bot}} \in {\cal N}P_{_{W(X)^\bot}}$ and $x\in {\cal H}$. Then $(N P_{_{W(X)^\bot}})(WP_{_X})(x) = (N P_{_{W(X)^\bot}})(y)$ where $y = (WP_{_X})(x) \in W(X)$. Thus $P_{_{W(X)^\bot}}(y) = 0$ and $(N P_{_{W(X)^\bot}})(WP_{_X})(x) = 0$. It proves that ${\cal N}P_{_{W(X)^\bot}} \subseteq Ann_R(WP_{_X})$.  For the other inclusion let us consider $S \in Ann_R(WP_{_X})$. In this way $S\upharpoonright_{W(X)} = 0$.  By Proposition \ref{PARCISOM1}-1, $W(X)$ is a closed subspace of ${\cal H}$ and then, for each $x\in {\cal H}$ we can decompose $x$ as $x = x_{_{W(X)}} + x_{_{W(X)^\bot}}$ where $x_{_{W(X)}} \in W(X)$ and $x_{_{W(X)^\bot}} \in W^\bot(X)$. Thus, $S(x) = S(x_{_{W(X)}}) + S(x_{_{W(X)^\bot}}) = 0 + S(x_{_{W(X)^\bot}}) = SP_{_{W(X)^\bot}}(x)$ and then $S = SP_{_{W(X)^\bot}} \in {\cal N}P_{_{W(X)^\bot}}$.  It proves that  $Ann_R(WP_X) \subseteq {\cal N}P_{_{W(X)^\bot}}$. Hence, $P_{_{W(X)^\bot}} = (WP_X)' \in {\cal P}({\cal N})$ and, by Proposition \ref{SUMPROJ}, $P_{_{W(X)}} = 1_{{\cal H}} - P_{_{W(X)^\bot}} \in {\cal P}({\cal N})$.

\qed
\end{proof}

Let us notice that Proposition \ref{PARCISOMCLOSED} suggests that each partial isometry $W$ in a von Neumann algebra ${\cal N}\subseteq  {\cal B}({\cal H})$ satisfying $WW^* = 1_{{\cal H}}$ defines, in a natural way,
two operations on the lattice ${\cal P}({\cal N})$ given by
\begin{equation}\nonumber
{\cal P}({\cal N}) \ni P_{_X} \longmapsto P_{_{W(X)}} \hspace{0.3cm} \mbox{and} \hspace{0.3cm} {\cal P}({\cal N}) \ni P_{_X} \longmapsto P_{_{W^*(X)}}.
\end{equation}
For the sake of simplicity and in order to study these operations will be more convenient to work with closed subspaces rather than projectors.
Then, by considering the usual identification $P_{_X} \cong X $ for elements of ${\cal P}({\cal N})$ the following definition formally introduce the above operations.

\begin{definition}\label{WOPERATIONS}
{\rm Let ${\cal N} \subseteq  {\cal B}({\cal H})$ be a von Neumann algebra and $W \in {\cal N}$ be a partial isometry such that $WW^* = 1_{{\cal H}}$. By denoting $Z = Ker^\bot(W)$ we define the pair of {\it Borchers operations associated} to $W$ as the operations $\{w_Z, w^*_Z\}$ on ${\cal P}({\cal N})$ given by
\begin{equation}\label{OPW}
{\cal P}({\cal N}) \ni P_{_X} \cong X \longmapsto w_Z(X) = W(X) \cong  P_{_{W(X)}}.
\end{equation}
\begin{equation}\label{OPWSTAR}
{\cal P}({\cal N}) \ni P_{_X} \cong X \mapsto w^*_Z(X) = W^*(X) \cong P_{_{W^*(X)}}
\end{equation}
}
\end{definition}

\begin{prop}\label{PARCISOMID}
Let ${\cal N} \subseteq {\cal B}({\cal H})$ be a von Neumann algebra and $W \in {\cal N}$ be a partial isometry such that $WW^* = 1_{\cal H}$ and $Z = Ker^\bot(W)$. Then, the Borchers operations $w_Z$ and $w^*_Z$ associated to $W$ satisfy the following:
\begin{enumerate}
\item
$w_Z (w_Z^* (X)) = X$ and $w_Z^*(w_Z (X)) = \mu_{_Z}(X)$.

\item
$w_Z$ defines an order preserving operation on ${\cal P}({\cal N})$ such that the restriction $w_Z\upharpoonright_{[0, Z]_{{\cal P}({\cal N})}}$ is a ${\cal OML}$-isomorphism onto ${\cal P}({\cal N})$.

\item
$w^*_Z$ is a ${\cal OML}$-isomorphism of the form $w^*_Z: {\cal P}({\cal N}) \rightarrow [0, Z]_{{\cal P}({\cal N})}$.

\end{enumerate}

\end{prop}

\begin{proof}
1) They follow by the condition  $WW^* = 1_{\cal H}$ and by  Proposition \ref{EQISOM}-3 respectively.

2) We first note that $w_Z$ is an order preserving operation on ${\cal P}({\cal N})$. Then, by Proposition \ref{EXTEN}-2, the restriction $w_Z\upharpoonright_{[0, Z]_{{\cal P}({\cal N})}}$ is an order isomorphism onto ${\cal P}({\cal N})$. We have to prove that $w_Z$ preserves orthogonal complements i.e.,  for each $X \in [0, Z]_{{\cal P}({\cal N})}$, $w_Z(\neg_{_Z}X) = w_Z(X)^\bot$. Indeed:
Let $y \in w_Z(\neg_{_Z}X)$. Then  there exists $x \in   \neg_{_Z}X = X^\bot \cap Z $ such that $y = W(x)$. Let $y_1 \in w_Z(X)$. Then there exists $x_1 \in X$ such that $y_1 = W(x_1)$. Let us notice that $\langle x, x_1 \rangle = 0$ and $x, x_1 \in Z = Ker^\bot(W)$. Thus, by Proposition \ref{PARCISOM0}-2, $\langle y, y_1 \rangle = \langle W(x), W(x_1) \rangle = \langle x, x_1 \rangle = 0$. It proves that $y \bot w_Z(X)$ and consequently $y\in w_Z(X)^\bot$. Therefore, $w_Z(\neg_{_Z}X) \subseteq w_Z(X)^\bot$. For the converse, let $y \in  w_Z(X)^\bot$. Then, for each $z\in X$, we have that $\langle y, W(z) \rangle = 0$.
By Proposition \ref{PARCISOM0}-2 there exists $x \in Z$ such that $y = W(x)$. Let us notice that $x \bot X $ because $\langle x,z \rangle = \langle W(x), W(z) \rangle =  \langle y, W(z) \rangle = 0$ for each $z\in X$.
Thus, $x\in X^\bot \cap Z$ and $y = W(x) \in w_Z(X^\bot \cap Z) = w_Z(\neg_{_Z}X) $. It proves that $w_Z(X)^\bot \subseteq  w_Z(\neg_{_Z}X)$. Hence,
the restriction $w_Z\upharpoonright_{[0, Z]_{{\cal P}({\cal N})}}$ preserves orthogonal complements and then, $w_Z\upharpoonright_{[0, Z]_{{\cal P}({\cal N})}}$ is a ${\cal OML}$-isomorphism onto ${\cal P}({\cal N})$.

3) Let us notice that $w^*_z$ defines an order preserving operation on ${\cal P}({\cal N})$. Then, by Proposition \ref{EXTEN}-1, $w^*_Z$ is an order isomorphisms onto $[0, Z]_{{\cal P}({\cal N})}$. We have to prove that $w_Z^*$ preserves orthogonal complements i.e., for each $X \in {\cal P}({\cal N})$, $w_Z^*(X^\bot) = \neg_{_Z} w_Z^*(X)$. Indeed: We first note that $w_Z\big(w_Z^*(X^\bot)\big) = X^\bot$ because of item 1.
We also note that $ w_Z^*(X) \in [0, Z]_{{\cal P}({\cal N})}$. Then, by item 2, $w_Z \big(\neg_{_Z} w_Z^*(X) \big) = \big( w_Z w_Z^*(X)\big)^\bot = X^\bot$. Thus, $w_Z\big(w_Z^*(X^\bot)\big) = w_Z \big(\neg_{_Z} w_Z^*(X) \big)$ and
, since the restriction $w_Z\upharpoonright_{[0, Z]_{{\cal P}({\cal N})}}$ is bijective, we have that $w_Z^*(X^\bot) = \neg_{_Z} w_Z^*(X)$.
Hence $w_Z^*$ preserves orthogonal complements and it is a ${\cal OML}$-isomorphism from ${\cal P}({\cal N})$ onto $[0, Z]_{{\cal P}({\cal N})}$.

\qed
\end{proof}

As we will see below, Proposition \ref{PARCISOMID} provides a useful result in order to establish an alternative proof of the well known properties of non-modularity and non-atomicity of the lattice of projectors of a type $III$ factor.

\begin{prop}\label{RELATED2}
Let ${\cal N} \subseteq {\cal B}({\cal H})$ be a  type III factor. Then

\begin{enumerate}
\item
If $0< Y < X < 1_{{\cal H}}$ then $Y \sim_p X$.

\item
${\cal P}({\cal N})$ is a non modular lattice.

\item
${\cal P}({\cal N})$ has no atoms.

\end{enumerate}

\end{prop}

\begin{proof}
1) By Proposition \ref{RELATED1}-2, we have that  $X \sim Y < X $. Moreover, by Proposition \ref{PARCISOMID}, there exists a partial isometry defining an order isomorphism  $w^*_{X^\bot}: {\cal P}({\cal N}) \rightarrow [0, {X^\bot}]_{{\cal P}({\cal N})}$. Therefore, $w^*_{X^\bot}(X^\bot) \leq X^\bot$ and, by Proposition \ref{RELATED1}-2 again, we also have that $Y \preceq w^*_{X^\bot}(X^\bot)$.  In this way $Y \preceq w^*_{X^\bot}(X^\bot) \leq X^\bot$. Then, by Proposition \ref{FILLLEMMA}, $Y$ and $X$ are perspective in $[0, X \lor w^*_{X^\bot}(X^\bot) ]_{{\cal P}({\cal N})}$. Thus, by Proposition \ref{EQPESPEC}, $Y \sim_p X$.

2) Immediately follows from item 1.

3) Let $Z \in{\cal P}({\cal N}) - \{0, 1_{{\cal H}}\}$. By Proposition \ref{RELATED1} there exists a partial isometry $W$ such that $WW^* = 1_{\cal H}$ and $W^*W = P_{Z^\bot}$. Thus, by Proposition \ref{PARCISOMID}-3, we have that $0 < w^*_Z(Z) < Z$ and, consequently, $Z$ is not an atom. Hence ${\cal P}({\cal N})$ has no atoms.

\qed
\end{proof}

The following theorem provides a purely lattice order characterization of type III factors.

\begin{theo}\label{CHARACTTYPEIII}
Let ${\cal N} \subseteq {\cal B}({\cal H})$ be factor. Then the following  statements are equivalent:

\begin{enumerate}

\item
${\cal N}$ is a a non trivial type III factor.

\item
${\cal P}({\cal N})$ admits two binary operations $w(-, -)$ and $w^*(-,-)$ such that, upon defining $w_{_Z}(X) = w(Z,X)$ and $w_{_Z}^*(X) = w^*(Z,X)$, the following conditions are satisfied:

\begin{enumerate}[\hspace{18pt}{\rm {\footnotesize III}}1.]

\item
$w_0(0) = 0$,

\item
$X\leq w_0(X)$,

\item
$Y = (Y\land w_0(X)) \lor (Y\land w_0(X)^\bot)$,

\item
$w_{_Z}(X\land Y) \leq w_{_Z}(X) $,

\item
$w_0(Z) \land w_{_Z}^* (X \land Y)  \leq w_{_Z}^* (X) $,

\item $\begin{aligned}[t]
    w_0^*(Z) \land Z &= w_0^*(Z) \land w_{_Z}^*(Z) \\
              &= \Big(w_0(Z^\bot)^\bot \land  w_0^*(Z) \Big) \lor \Big(w_0(Z^\bot) \land w_0\big( w_0^*(Z) \land Z \big) \Big),
  \end{aligned}$

\item
$w_0^*(Z) \lor Z =  w_0^*(Z) \lor w_{_Z}^*(Z) = w_0\big( w_0^*(Z) \lor Z \big) $,

\item
$w_0(Z) \leq w_{_Z}\big(w_{_Z}^* (X)\big) R X  $,

\item
$w_0(Z) \leq  w_{_Z}^*\big( w_{_Z}(X)\big) R \mu_{_Z}(X)  $,

\item
$w_0\big( w_0^*(1)\big) = w_0\big( w_0^*(1)^\bot \big)$.

\end{enumerate}

In this case  $w_0$ is the internal dimension function on ${\cal P}({\cal N})$ and $w_0^*$ satisfies the following conditions
\begin{equation}\label{w_0^*(0)}
w_0^*(0) = 0.
\end{equation}
\begin{equation}\label{w_0^*(1)}
0< w_0^*(1_{\cal H}) <1_{\cal H}.
\end{equation}
\begin{equation}\label{w_0^*(Z)common}
\mbox{$w_0^*(Z)$ is a common complement of $\{Z, w_{_Z}^*(Z) \}$ for $Z \not = 0, 1_{\cal H}$.}
\end{equation}

\end{enumerate}

\end{theo}

\begin{proof}
$1 \Longrightarrow 2)$ Let us suppose that ${\cal N}$ is a non trivial type III factor. Let $w_0$ be the internal dimension function on ${\cal P}({\cal N})$. Then, by Proposition \ref{BAAZ}, conditions {\footnotesize III}{\it 1},  {\footnotesize III}{\it 2} and {\footnotesize III}{\it 3} are satisfied.
By the Borchers condition, mentioned in Proposition \ref{RELATED1}, for each $Z \in {\cal P}({\cal N}) - \{0\}$ there exists a partial isometry $W_Z$ such that $W_ZW^*_Z = 1_{\cal H}$ and $W^*_Z W_Z = P_{_Z}$, thereby defining a pair of Borchers operations $\{w_Z, w^*_Z\}$ on $ {\cal P}({\cal N})$ associated to $W_Z$. By Proposition \ref{PARCISOMID}, $w_Z$ and $w^*_Z$ are order preserving operations such that
\begin{equation}\label{EQWZN0}
w_Z (w_Z^* (X)) = X, \hspace{0.3cm} w_Z^*(w_Z (X)) = \mu_{_Z}(X).
\end{equation}
In this way, by Proposition \ref{PARCISOMID}-3, for each $Z \in {\cal P}({\cal N}) -\{0, 1_{\cal H}\}$ we have that $0 < w_{_Z}^*(Z) < Z$ and, by Proposition \ref{RELATED2}-1, $w_{_Z}^*(Z) \sim_p Z$. This allow us to define
an unary operation $w_0^*(-)$ on ${\cal P}({\cal N})$ such that
\begin{itemize}
\item[]
$w_0^*(0) = 0$,
\item[]
$0< w_0^*(1_{\cal H}) <1_{\cal H}$,
\item[]
$w_0^*(Z)$ is a common complement of $\{Z, w_{_Z}^*(Z) \}$ for $Z \not = 0, 1_{\cal H}$.
\end{itemize}
In this way $w_0^*$ satisfies Eq.(\ref{w_0^*(0)}), Eq.(\ref{w_0^*(1)}) and Eq.(\ref{w_0^*(Z)common}). Moreover, $w_0\big( w_0^*(1)\big) = w_0\big( w_0^*(1)^\bot \big)$ because $w_0^*(1) \not = 0, 1_{\cal H}$. Thus, condition {\footnotesize III}{\it 10} is also satisfied.

{\footnotesize III}{\it 4}. Let us notice that $w_0$ is an order preserving map. Then, $w_0(X\land Y) \leq w_0(X)$. If $Z\not=0$ then the partial isometry $W_Z$ preserves the inclusion of closed subspaces. Thus, $w_{_Z}(X\land Y) \leq w_{_Z}(X) $.

{\footnotesize III}{\it 5}. If $Z=0$ then $w_0(0) \land w_0^* (X \land Y) = 0 \leq  w_0^* (X)$. Let us suppose that $Z\not=0$. Then $W_Z^*$ preserves the inclusion of closed subspaces because it is also a partial isometry.
Thus, $w_0(Z) \land w_{_Z}^* (X \land Y) = 1_{\cal H} \land w_{_Z}^* (X \land Y)   \leq w_{_Z}^* (X) $.

{\footnotesize III}{\it 6}. By straightforward calculation we can see that condition {\footnotesize III}{\it 6} is satisfied for $Z \in \{0, 1_{\cal H}\} $. Let us suppose that $0 < Z < 1_{\cal H}$. By definition of $w_0^*$,
$w_0^*(Z)$ is a common complement of $\{Z, w_{_Z}^*(Z) \}$ and then, $w_0^*(Z) \land Z = w_0^*(Z) \land w_{_Z}^*(Z) = 0$. Let also notice that $0 < Z^\bot < 1_{\cal H}$. Then, $\Big(w_0(Z^\bot)^\bot \land  w_0^*(Z) \Big) \lor \Big(w_0(Z^\bot) \land w_0\big( w_0^*(Z) \land Z \big) \Big) = \big(0 \land  w_0^*(Z)\big) \lor \big(1 \land w_0\big( w_0^*(Z) \land Z \big)\big) = 0 \lor (1 \land w_0 (0)) = 0 $. Thus {\footnotesize III}{\it 6} is also satisfied in this case.

{\footnotesize III}{\it 7}. If $Z \in \{0, 1_{\cal H}\} $ then, by straightforward calculation, {\footnotesize III}{\it 7} is satisfied. The case $0 < Z < 1_{\cal H}$ follows along the same lines as the proof of condition
{\footnotesize III}{\it 6}.

{\footnotesize III}{\it 8}, {\footnotesize III}{\it 9}. If $Z=0$ then {\footnotesize III}{\it 8} and {\footnotesize III}{\it 9} are trivially satisfied. If $Z \not = 0$ then, by Eq.(\ref{EQWZN0}), $w_{_Z}(w_{_Z}^* (X)) R X = 1$ and $w_{_Z}^*( w_{_Z}(X)) R \mu_{_Z}(X) = 1$. Thus {\footnotesize III}{\it 8} and {\footnotesize III}{\it 9} are also satisfied in this case.

Finally, if we define the binary operations $w(Z, X) = w_Z(X)$ and $w^*(Z,X) = w^*_Z(X)$ then our claim is established.

$2 \Longrightarrow 1)$ Let us suppose that ${\cal P}({\cal N})$ admits two binary operations $w(-, -)$ and $w^*(-,-)$ such that, upon defining $w_{_Z}(X) = w(Z,X)$ and $w_{_Z}^*(X) = w^*(Z,X)$, conditions {\footnotesize III}{\it 1} $\ldots$ {\footnotesize III}{\it 10}  are satisfied. Since  ${\cal N}$ is a factor, combining conditions {\footnotesize III}{\it 1}, {\footnotesize III}{\it 2}, {\footnotesize III}{\it 3} and Proposition \ref{BAAZ},  $w_0(0) = 0$ and $w_0(X) = 1_{\cal H}$ whenever $X \not= 0 $.  By condition {\footnotesize III}{\it 10} we have that $0< w_0^*(1_{\cal H}) <1_{\cal H}$. Thus, Eq.(\ref{w_0^*(1)}) is satisfied and ${\cal N}$ is a non trivial factor.
By condition {\footnotesize III}{\it 6} we have that $0 = w_0^*(0) \land 0 = w_0^*(0) \land w_0^*(0) = w_0^*(0)$. Therefore, Eq.(\ref{w_0^*(0)}) is satisfied. Let us suppose that $0 < Z < 1_{\cal H}$. Then $w_0(Z) = w_0(Z^\bot) = 1$ and, by condition {\footnotesize III}{\it 6},  $w_0^*(Z) \land Z = w_0^*(Z) \land w_{_Z}^*(Z) = w_0\big( w_0^*(Z) \land Z \big) \in \{ 0, 1_{\cal H} \}$. If $w_0^*(Z) \land Z = 1_{\cal H}$ then $Z = 1_{\cal H}$ which is a contradiction. Then,
\begin{equation}\label{w_0^*(Z)common1}
w_0^*(Z) \land Z = w_0^*(Z) \land w_{_Z}^*(Z) = 0.
\end{equation}
By condition {\footnotesize III}{\it 7}, $w_0^*(Z) \lor Z =  w_0^*(Z) \lor w_{_Z}^*(Z) = w_0\big( w_0^*(Z) \lor Z \big) \in \{ 0, 1_{\cal H} \}$. If $w_0^*(Z) \lor Z = 0$ then $Z=0$ which is a contradiction. Then,
\begin{equation}\label{w_0^*(Z)common2}
w_0^*(Z) \lor Z =  w_0^*(Z) \lor w_{_Z}^*(Z) = 1.
\end{equation}
Thus, by Eq.(\ref{w_0^*(Z)common1}) and Eq.(\ref{w_0^*(Z)common2}), $w_0^*(Z)$ is a common complement of $\{Z, w_{_Z}^*(Z) \}$ whenever $Z \not = 0, 1_{\cal H}$. Consequently, Eq.(\ref{w_0^*(Z)common}) is also satisfied.

By Theorem \ref{TYPECLASSIFICATION}, there exists a dimension function $D:{\cal P}({\cal N}) \rightarrow [0,\infty]$ uniquely determined up to positive
constant multiple. Then we shall prove that for each $Z \in {\cal P}({\cal N}) -\{0, 1_{\cal H} \}$, $D(Z) = \infty$. Indeed:

Since $Z\not=0$ we have that $w_0(Z) = 1$. Then, by conditions {\footnotesize III}{\it 8} and {\footnotesize III}{\it 9}, $ w_{_Z}(w_{_Z}^* (X)) R X = 1$ and $w_{_Z}^*( w_{_Z}(X)) R \mu_Z(X) = 1$ respectively.
Therefore, by Eq.(\ref{ECMOL}), we have that
\begin{equation}\label{W1}
w_{_Z}(w_{_Z}^* (X)) = X \hspace{0.2cm} and \hspace{0.2cm} w_{_Z}^*( w_{_Z}(X)) = \mu_Z(X).
\end{equation}
We also note that, by condition {\footnotesize III}{\it 4} and {\footnotesize III}{\it 5}, $w_{_Z}$ and $w_{_Z}^*$ are order preserving maps in ${\cal P}({\cal N})$. Thus, by Eq.(\ref{W1}) and  Proposition \ref{EXTEN}, we have that
\begin{equation}\label{W2}
0 < w_{_Z}^*(Z) <Z <1.
\end{equation}
We also note that, by Eq.(\ref{w_0^*(Z)common1}) and Eq.(\ref{w_0^*(Z)common2}),  $ w_{_Z}^*(Z) \sim_p Z$. Then, by Eq.(\ref{W2}) and  by Proposition \ref{INFINITEDIM}-2, we have that $D(Z) = \infty$. It proves that ${\cal N}$ is a type III factor and $w_0$ is the internal dimension function.

\qed
\end{proof}

Let us notice that conditions  {\footnotesize III}{\it 1} $\ldots$  {\footnotesize III}{\it 10} can be rephrased as an equational system because ${\cal P}({\cal N})$ is a lattice.

\begin{prop}\label{PRECONGCOND}
Let ${\cal N} \subseteq {\cal B}({\cal H})$ be a type III factor and let us consider two binary operations $w(-, -)$ and $w^*(-,-)$ on ${\cal P}({\cal N})$
satisfying the conditions {\footnotesize III}1 $\ldots$ {\footnotesize III}10 of the Theorem \ref{CHARACTTYPEIII}. Then,

\begin{enumerate}

\item
$w_{_X}(y\land w_0(Z)) = w_{_{X\land w_0(Z)}}(Y\land w_0(Z)) = w_{_X}(Y) \land w_0(Z)$.

\item
$w^*_{_X}(Y\land w_0(Z)) = w^*_{_{X\land w_0(Z)}}(Y\land w_0(Z)) = w^*_{_X}(Y) \land w_0(Z)$.

\end{enumerate}

\end{prop}

\begin{proof}
1) Let us suppose that $Z=0$. Then, $w_{_X}(Y\land w_0(0)) = w_{_X}(0) = 0$, $w_{_{X\land w_0(0)}}(Y\land w_0(0)) = w_0(Y\land 0) = 0$ and $w_{_X}(Y) \land w_0(0) = w_{_X}(Y) \land 0 = 0$. Thus, the equations are satisfied in this case. Let us suppose that $Z \not= 0$. Then $w_{_X}(y\land w_0(Z)) = w_{_X}(Y\land 1) = w_{_X}(Y)$, $w_{_{X\land w_0(Z)}}(Y\land w_0(Z)) = w_{_{X\land 1}}(Y\land 1) = w_{_X}(Y)$ and $ w_{_X}(Y) \land w_0(Z) = w_{_X}(Y) \land 1 = w_{_X}(Y)$. Thus the equations are also satisfied in this case.

2) Let us suppose that $Z=0$. Then, by Eq.(\ref{w_0^*(0)}) and Proposition \ref{EXTEN}-1, $w^*_{_X}(Y\land w_0(0)) = w^*_X(0) = 0$, $w^*_{_{X\land w_0(0)}}(Y\land w_0(0)) = w^*_0(0) = 0$ and  $w^*_{_X}(Y) \land w_0(0) = 0$. Thus, the equations are satisfied in this case. Let us suppose that $Z \not= 0$. Then, $w^*_X(Y\land w_0(Z)) = w^*_X(Y\land 1) = w^*_X(Y)$,  $w^*_{X\land w_0(Z)}(Y\land w_0(Z)) = w^*_{X\land 1}(Y\land 1) = w^*_X(Y) $ and $w^*_X(Y) \land w_0(Z) = w^*_X(Y) \land 1 = w^*_X(Y)$. Thus the equations are also satisfied in this case.

\qed
\end{proof}

\section{The algebraic model of $LQF$-logic}\label{ALMODELQF}

Let us consider a net of observables $\mathfrak{R}({\cal M})\ni {\cal O} \mapsto {\cal N}({\cal O}) $ over the spacetime ${\cal M}$. As we have already mentioned in Section \ref{RINGOP}, each local observable algebra ${\cal N}({\cal O})$ defines a local propositional system encoded in the von Neumann lattice ${\cal P}({\cal N}({\cal O}))$. Furthermore, as shown in Theorem \ref{CHARACTTYPEIII}, a set of equation formulated in an expanded language of ${\cal P}({\cal N}({\cal O}))$ is able to capture the fundamental requirement of the net of observables, namely, the type III factor condition of each algebra of the net. Thus, we are led to the following extension of the orthomodular structure that defines the algebraic model for the $LQF$-logic.

\begin{definition}\label{LQFlatticeDef}
{\rm A {\it $LQF$-algebra} is an algebra $\langle A, \land, \lor, w, w^*, \neg, 0,1 \rangle$ of type  $\langle 2,2,2,2,1,0,0 \rangle$ such that, upon defining $w_z(x) = w(z,x)$ and $w_z^*(x) = w^*(z,x)$, the following conditions are satisfied:
\begin{enumerate}[\hspace{18pt}{\rm LQF}1.]

\item[{\rm LQF}0.]
$\langle A, \land, \lor, \neg, 0,1 \rangle$ is an orthomodular lattice,

\item
$w_0(0) = 0$,

\item
$x\leq w_0(x)$,

\item
$y = (y\land w_0(x)) \lor (y\land \neg w_0(x))$,

\item
$w_z(x\land y) \leq w_z(y)$,

\item
$w_0(z) \land w_z^* (x \land y)  \leq w_z^* (x) $,

\item $\begin{aligned}[t]
    w_0^*(z) \land z &= w_0^*(z) \land w_z^*(z) \\
              &= \Big(\neg w_0( \neg z) \land  w_0^*(z) \Big) \lor \Big(w_0(\neg z) \land w_0\big( w_0^*(z) \land z \big) \Big),
  \end{aligned}$

\item
$w_0^*(z) \lor z =  w_0^*(z) \lor w_z^*(z) = w_0\big( w_0^*(z) \lor z \big) $,

\item
$w_0(z) \leq w_z\big(w_z^* (x)\big) R x  $,

\item
$w_0(z) \leq  w_z^*\big( w_z(x)\big) R \mu_z(x)  $,

\item
$w_0\big( w_0^*(1)\big) = w_0\big(\neg w_0^*(1) \big)$,

\item
$w_x(y\land w_0(z)) = w_{x\land w_0(z)}(y\land w_0(z)) = w_x(y) \land w_0(z)$,

\item
$w^*_x(y\land w_0(z)) = w^*_{x\land w_0(z)}(y\land w_0(z)) = w^*_x(y) \land w_0(z)$.

\end{enumerate}
}
\end{definition}

In the same way as in Theorem \ref{CHARACTTYPEIII} conditions {\rm LQF}0 $\ldots$ {\rm LQF}12 can be rephrased as an equational system. Thus, the class of $LQF$-algebras defines a variety of algebras that we denote by ${\cal LQF}$. By a $LQF$-homomorphism we mean a $\langle \land, \lor, w, w^*, \neg, 0,1 \rangle$-preserving function between $LQF$-algebras.

\begin{rema}
{\rm Let us notice that LQF8 and LQF9 allow us to represent the Borchers condition, that characterize the type III factor condition, in a purely algebraic way. }
\end{rema}

\begin{example}
{\rm Let ${\cal N} \subseteq {\cal B}({\cal H})$ be a non trivial type $III$ factor. From Theorem \ref{CHARACTTYPEIII} and Proposition  \ref{PRECONGCOND}, it immediately follows that the lattice of projector ${\cal P}({\cal N})$
defines a $LQF$-algebra.
}
\end{example}

\begin{prop}\label{ARITLQF}
Let $A$ be a $LQF$-algebra. Then

\begin{enumerate}
\item
$w_0 (1) = 1$.

\item
$w_0^* (0) = 0$.

\item
$w_0 (x) \in Z(A)$.

\item
If $x\leq y$ then $w_z(x) \leq w_z(y)$.

\item
If $z<1$ and $w_0(z) = 1$ then $0 < w_z^*(z) < z$.

\end{enumerate}
\end{prop}

\begin{proof}
1) It follows from LQF2. 2) If $z= 0$ then, by LQF1 and LQF12, we have that $w_0^* (0) = w^*_0(0\land w_0(0)) = w^*_0(0) \land w_0(0) = w^*_0(0) \land 0 = 0$.  3) It follows from LQF3 and Proposition \ref{eqcentro}-2.
4) If $x\leq y$ then, by LQF4, $w_z(x) = w_z(x\land y) \leq w_z(y)$. 5) Let us suppose that $z<1$ and $w_0(z) = 1$. By LQF5 $w_z^*$ is an order preserving map and, by LQF8 and LQF9, $w_z w_z^*  = id_A  $ and $ w_z^*  w_z = \mu_z$ respectively. Thus, by Proposition \ref{EXTEN}-1, $w_z^*: A \rightarrow [0,z]$ is an order isomorphism. Consequently, $0 < w_z^*(z) < w_z^*(1) = z$ because $0< z <1$. 6) It follows from LQF7 and item 3.

\qed
\end{proof}

\begin{prop}\label{LQFFACTOR}
Let $A$ be a $LQF$-algebra. Then, for each $z\in A$, the relation $\theta_{w_0(z)}$ on $A$ given by $$(a,b) \in \theta_{w_0(z)} \hspace{0.3cm} \mbox{iff} \hspace{0.3cm} a\land w_0(z) = b\land w_0(z)$$ is a factor congruence and $A/_{\theta_{w_0(z)}} \simeq_{_{LQF}} [0, w_0(z)]_A$.

\end{prop}

\begin{proof}
By Proposition \ref{ARITLQF}-3 we have that $w_0(z) \in Z(A)$. Then $\theta_{w_0(z)}$ is a factor ${\cal OML}$-congruence and, by Proposition \ref{DIOML}, $A/_{\theta_{w_0(z)}} \simeq_{_{OML}} [0, w_0(z)]_A$.  We have to prove that $w(-,-)$ and $w^*(-,-)$ are compatible operations with respect to $\theta_{w_0(z)}$. Indeed: Let us suppose that $$a\land w_0(z) = b\land w_0(z).$$
In order to prove the compatibility of $w$ let us notice that, by Axiom LFQ11, for each $x\in A$  we have that

$\begin{aligned}[t]
    w(x, a) \land w_0(z)  &= w_x(a)\land w_0(z) = w_x(a\land w_0(z)) = w_x(b\land w_0(z)) \\
                          &= w_x(b)\land w_0(z) = w(x, b)\land w_0(z)
\end{aligned}$

\noindent and

$\begin{aligned}[t]
    w(a, x) \land w_0(z)  &=   w_a(x)\land w_0(z) =  w_{a\land w_0(z)}(x \land w_0(z))  \\
                          &=  w_{b\land w_0(z)}(x \land w_0(z)) = w_b(x) \land w_0(z) \\
                          &=  w(b, x) \land w_0(z).
\end{aligned}$

\noindent
The above equations then give immediately the compatibility of $w$ with respect to $\theta_{w_0(z)}$.

For the compatibility of $w^*$ let us notice that, by Axiom LFQ12, for each $x\in A$  we have that

$\begin{aligned}[t]
    w^*(x, a) \land w_0(z)  &= w^*_x(a)\land w_0(z) = w^*_x(a\land w_0(z)) = w^*_x(b\land w_0(z)) \\
                          &= w^*_x(b)\land w_0(z) = w^*(x, b)\land w_0(z)
\end{aligned}$

\noindent and

$\begin{aligned}[t]
    w^*(a, x) \land w_0(z)  &=   w^*_a(x)\land w_0(z) =  w^*_{a\land w_0(z)}(x \land w_0(z)) \\
                            &=  w^*_{b\land w_0(z)}(x \land w_0(z)) =  w^*_b(x) \land w_0(z) \\
                            &= w^*(b, x) \land w_0(z).
\end{aligned}$

\noindent
The above equations then give immediately the compatibility of  $w^*$ with respect to $\theta_{w_0(z)}$.

Thus, the orthomodular lattice $[0, w_0(z)]_A$ endowed with the operations $$[0, w_0(z)]_A \ni (x,y) \mapsto w(x,y)_{/\theta_{w_0(z)}} = w(x,y) \land w_0(z)$$
$$[0, w_0(z)]_A \ni (x,y) \mapsto w^*(x,y)_{/\theta_{w_0(z)}} = w^*(x,y) \land w_0(z)$$ is a $LQF$-algebra and $A/_{\theta_{w_0(z)}} \simeq_{_{LQF}} [0, w_0(z)]_A$.

\qed
\end{proof}

\begin{prop}\label{DIRECTID}
Let $A$ be a $LQF$-algebra. Then the following conditions are equivalent:

\begin{enumerate}
\item
$A$ is directly indecomposable algebra.

\item
$w_0(z) = 1$ iff $z\not= 0$.

\item
For each $z \in A - \{0,1\}$, $w_0^*(z)$ is a common complement of  $\{w_z^*(z), z\}$.

\item
$Z(A) = \{0,1\}$.

\end{enumerate}

\end{prop}

\begin{proof}
$1 \Longrightarrow 2$) If $A$ is directly indecomposable algebra then, by Proposition \ref{LQFFACTOR}, $w_0(z) \in \{0,1\}$. Thus, by Axiom LFQ1 and LQF2, we have that  $w_0(z) = 1$ iff $x\not= 0$.

$2 \Longrightarrow 3$) Let $z \in A - \{0,1\}$. Then $0 < \neg z <1$ and, by hypothesis, $w_0(z) = w_0(\neg z) = 1$. Therefore, by LFQ6, we have that
$$\begin{aligned}[t]
    w_0^*(z) \land z &= w_0^*(z) \land w_z^*(z) = \Big(\neg 1 \land  w_0^*(z) \Big) \lor \Big(1 \land w_0\big( w_0^*(z) \land z \big) \Big) \\
              &= w_0( w_0^*(z) \land z) \in \{0,1\}.
  \end{aligned}$$
In this way $w_0^*(z) \land z = 0$; otherwise $z=1$, which is a contradiction. Thus,
\begin{equation}\label{1complement0}
w_0^*(z) \land z = w_0^*(z) \land w_z^*(z) = 0.
\end{equation}
By hypothesis and LFQ7 we also have that $w_0^*(z) \lor z =  w_0^*(z) \lor w_z^*(z) = w_0\big( w_0^*(z) \lor z \big) \in \{0,1\}$. In this way $w_0^*(z) \lor z = 1$; otherwise $z=0$, which is a contradiction. Thus,
\begin{equation}\label{1complement1}
w_0^*(z) \lor z = w_0^*(z) \lor w_z^*(z) = 1.
\end{equation}
Hence, by Eq.(\ref{1complement0}) and Eq.(\ref{1complement1}), $w_0^*(z)$ is a common complement of $\{z, w_z^*(z) \}$.

$3 \Longrightarrow 4$) Let us suppose that there exists $z \in Z(A)$ such that $0 < z <1$. By hypothesis we have that $w_0^*(z)$ is a common complement of $\{z, w_z^*(z) \}$. Let us notice that
$\neg z$ is the unique complement of $z$ because $z \in Z(A)$. Therefore $w_0^*(z) = \neg z$ and
\begin{equation}\label{1complement2}
z =w_z^*(z).
\end{equation}
Next, let us observe that $w_0(z) = 1$ then, by LFQ8 and LFQ9,  $w_z\big(w_z^* (x)\big) R x = 1$ i.e., $w_z\big(w_z^* (x)\big) = x$  and $w_z^*\big( w_z(x)\big) R \mu_z(x) = 1$ i.e., $w_z^*\big( w_z(x)\big) = \mu_z(x)$ respectively. Thus, by Proposition \ref{EXTEN}-1, $w_z^*(z) < z$ contradicting  Eq.(\ref{1complement2}). Hence, $Z(A) = \{0,1\}$.

$4 \Longrightarrow 1$) If $Z(A) = \{0,1\}$ then $A$ is directly indecomposable as orthomodular lattice. Then $A$ is directly indecomposable as $LQF$-algebra.

\qed
\end{proof}

\begin{prop}\label{ABSTRACTFACTOR}
Let $A$ be a $LQF$-algebra. Then:

\begin{enumerate}

\item
$A$ has no atoms.

\item
$Card(A) \geq \aleph_0$.

\item
$A$ is a non modular lattice.

\end{enumerate}
\end{prop}

\begin{proof}
1) We first prove that $1\in A$ is not an atom. Let us notice that $w_0^*(1) \not \in \{0,1\}$; otherwise, by Proposition \ref{DIRECTID}-2,  $w_0(w_0^*(1)) \not = w_0(\neg w_0^*(1))$, that contradicts LQF10. Thus, $0 < w_0^*(1) < 1$ and $1$ cannot be an atom. In other words $A \not = \{0,1\} $. Let $a \in A-\{0,1\}$.

First suppose that $A$ is a directly indecomposable $LQF$-algebra.  By Proposition \ref{DIRECTID}-2, $w_0(a) = 1$. Then, by Proposition \ref{ARITLQF}-5, we have that  $0 < w_a^*(a) < a$, so $a$ is not an atom.

For the general case, suppose that $a$ is an atom. Let us consider a subdirect representation $f: A \hookrightarrow \prod_{i\in I} A_i$ in the variety $\cal LQF$. Since $a>0$ and $f$ is injective then $f(a) = (a_i)_{i\in I} > 0$. Thus, there exists $j\in I$ such that $a_j = \pi_jf(a) >0$. Since $A_j$ is directly indecomposable algebra, $a_j$ is not an atom, so there exists $x_j \in A_j$ such that
\begin{equation} \label{xjlesaj}
0< x_j < a_j.
\end{equation}
Since $\pi_jf$ is a surjective $LQF$-homomorphism then there exists $x\in A-\{0\}$ such that $\pi_jf(x) = x_j$. Note that $x\not < a$ since we assumed that $a$ is an atom. Then we have to consider two possible cases:

If $a \leq x$ then $a_j = \pi_jf(a) \leq \pi_jf(x) = x_j$ that contradicts Eq.(\ref{xjlesaj}). Otherwise, if $a$ is not comparable to $x$ then, $a\land x = 0$ since $a$ is an atom. Thus, $0 = \pi_jf(a \land x) = \pi_jf(a) \land \pi_jf(x) = a_j \land x_j = x_j$ that also contradicts Eq.(\ref{xjlesaj}). Hence, $A$ has no atoms.

2) It immediately follows from the above item.

3) First let us suppose that $A$ is a directly indecomposable $LQF$-algebra.  By item 1 there exists $z \in A$ such that $0<z<1$ and then, by Proposition \ref{DIRECTID}-2, $w_0(z) = 1$. Therefore, by Proposition \ref{ARITLQF}-5, we have that $0 < w_z^*(z) < z$. Furthermore, by Proposition \ref{DIRECTID}-3, $w_0^*(z)$ is a common complement of $\{w_z^*(z), z\}$ and, consequently, $w_z^*(z) \sim_p z$.  Hence, by Proposition \ref{EQUIVNONMOD}-2, $A$ is a non modular lattice.

For the general case let us suppose that $A$ is a modular lattice. Let us consider a subdirect representation $f: A \hookrightarrow \prod_{i\in I} A_i$ in the variety $\cal LQF$. For $i\in I$ let us consider $z_i \in A_i$ such that $0<z_i<1$. Thus $0 < w_{z_i}^*(z) < z_i$ where
 $w_0^*(z_i)$ is a common complement of $\{w_{z_i}^*(z), z_i\}$ because $A_i$ is a directly indecomposable $LQF$-algebra. From this it follows that
\begin{equation}\label{NONMODDI1}
(w_{z_i}^*(z_i) \land z_i) \lor (w_0^*(z_i) \land z_i) \not = \big((w_{z_i}^*(z_i) \land z_i) \lor w_0^*(z_i) \Big) \land z_i.
\end{equation}
By the subjectivity of $\pi_if$ there exists $z\in A$ such that $\pi_if(z) = z_i$. Moreover, $(w_{z}^*(z) \land z) \lor (w_0^*(z) \land z) = \big((w_{z}^*(z) \land z) \lor w_0^*(z) \big) \land z$ since we assumed that
$A$ is a modular lattice. Thus,
$$\begin{aligned}[t]
   (w_{z_i}^*(z_i) \land z_i) \lor (w_0^*(z_i) \land z_i) &= \pi_if\big((w_{z}^*(z) \land z) \lor (w_0^*(z) \land z)\big) \\
              &= \pi_if\Big(\big((w_{z}^*(z) \land z) \lor w_0^*(z) \big) \land z \Big)\\
              &= \big((w_{z_i}^*(z_i) \land z_i) \lor w_0^*(z_i) \Big) \land z_i
  \end{aligned}$$
that contradicts Eq.(\ref{NONMODDI1}). Hence, $A$ is a non modular lattice.

\qed
\end{proof}

\begin{rema}\label{REMAABSTDI}
{\rm Combining  Theorem \ref{DIRECTID} and Proposition \ref{ABSTRACTFACTOR} we can observe that the class of $LQF$-algebras successfully captures the most
remarkable properties of the projector lattices of the type $III$ factors mentioned in Section \ref{RINGOP}.
}
\end{rema}

\begin{prop}\label{CENTRALCOVERGEN}
Let $A$ be a $LQF$-algebra. Then:

\begin{enumerate}
\item
For each $z\in Z(A)$, $w_0(z) = z$

\item
For each $a\in A$, $w_0(a) = e(a)$ and $\neg w_0( \neg a) = e_d(a)$.

\end{enumerate}

\end{prop}

\begin{proof}
1) Let $A$ be a $LQF$-algebra and, by Theorem \ref{Birkhoff}, let us consider a subdirect representation $f: A \hookrightarrow \prod_{i\in I} A_i$ in the variety $\cal LQF$.  Let $z\in Z(A)$. We first claim that, for each $i\in I$, $\pi_i f(z) \in Z(A_i)$. Indeed: Let $a_i \in A_i$. Since $\pi_i f$ is surjective then there exists $a\in A$ such that $\pi_i f(a) = a_i$. By Proposition \ref{eqcentro}-2 we have that $a = (a\land z) \lor (a\land \neg z)$ and therefore, $a_i = \pi_i(a) = \big(\pi_i(a) \land \pi_i(z) \big) \lor \big(\pi_i(a)\land \pi_i(\neg z)\big) = (a_i \land \pi_i(z)) \lor (a_i \land \neg \pi_i(z))$. Thus, again by Proposition \ref{eqcentro}-2, $\pi_i f(z) \in Z(A_i)$ as claimed.  Since $A_i$ is a directly indecomposable $LQF$-algebras, so $Z(A_i) = {\bf 2}$, then $f(z) = (z_i)_{i\in I}$ where $z_i \in {\bf 2}$. Thus, $f(w_0(z)) = w_0 (f(z)) = (w_0(z_i))_{i\in I} = (z_i)_{i\in I} = f(z)$. Hence, $z = w_0(z)$ because $f$ is an injective map.

2) Let $a\in A$ and let us consider the set $Z^\uparrow(a) = \{z\in Z(A): a\leq z \}$. By LFQ2 and Proposition \ref{ARITLQF}-3 we have that $w_0(a) \in Z^\uparrow(a)$. Let $z\in Z^\uparrow(a)$. Then, by Proposition \ref{ARITLQF}-4 and item 1, $w_0(a) \leq w_0(z) = z$. It proves that $w_0(a) = \bigwedge Z^\uparrow(a) = e(a)$. Finally, by Eq.(\ref{DUALCENTRALCOVER}), $\neg w_0( \neg a) = e_d(a)$.

\qed
\end{proof}

\begin{prop}\label{BOOLEANCOVERAUX}
Let $A$ be a $LQF$-algebra. Then,

\begin{enumerate}
\item
$e_d(e_d(x)) = e_d(x)$.

\item
If $x\leq y$ then $e_d(x) \leq e_d(y)$.

\item
$e_d(x \land y) = e_d(x) \land e_d(y)$.

\item
$e_d (x_1 R x_2) \land e_d (y_1 R y_2) \leq w(x_1, y_1) R w(x_2, y_2)$.

\item
$e_d (x_1 R x_2) \land e_d (y_1 R y_2) \leq w^*(x_1, y_1) R w^*(x_2, y_2)$.

\end{enumerate}
\end{prop}

\begin{proof}
1) It immediately follows from Proposition \ref{CENTRALCOVERGEN}.

2) If $x\leq y$ then  $e_d(x) \leq y$. Since $e_d(x) \in Z(A)$ then, $e_d(x) \leq \bigvee \{z\in Z(A): z\leq y \} = e_d(y)$.

3) We first note that $e_d(x) \land e_d(y) \in Z(A)$ and $e_d(x) \land e_d(y) \leq x\land y$. Let $z\in Z(A)$ such that $z\leq x\land y \leq x,y$. Then, by item 2 and Proposition \ref{CENTRALCOVERGEN}-1, $z = e_d(z) \leq e_d(x), e_d(y)$, so $z\leq e_d(x) \land e_d(y)$. Hence, $e_d(x) \land e_d(y) = \bigvee\{z\in Z(A): z\leq x\land y \} = e_d(x\land y)$.

4,5) Let us remark that the inequalities of these items can be equivalently formulated by equations. Then, by Eq.(\ref{EQDISI}), it is enough to study these inequalities in the class $\cal{DI}(\cal{LQF})$. Let $A$ be a directly indecomposable $LQF$-algebra. By Proposition \ref{ARITLQF}-3 and Proposition \ref{DIRECTID}-4 we have that $e_d (x_1 R x_2) \land e_d (y_1 R y_2) \in Z(A) = \{0,1\}$. Let us assume that $e_d (x_1 R x_2) \land e_d (y_1 R y_2) = 1$. Then, by Eq.(\ref{ECMOL}), $x_1 = x_2$ and $y_1 = y_2$ and so $w(x_1, y_1) = w(x_2, y_2)$, $w^*(x_1, y_1) = w^*(x_2, y_2)$. Again by Eq.(\ref{ECMOL}), we have that $w(x_1, y_1) R w(x_2, y_2) = 1$ and  $w^*(x_1, y_1) R w^*(x_2, y_2) = 1$. Hence our claim.

\qed
\end{proof}

\begin{prop}\label{DISC1}
Let $A$ be a directly indecomposable $LQF$-algebra. Then the term $$t(x,y,z) =  \Big(x \land \neg e_d (xRy)\Big) \lor  \Big(z \land e_d(xRy) \Big)$$ is a discriminator term for $A$ and ${\cal LQF}$ is a discriminator variety.
\end{prop}

\begin{proof}
Let $A$ be a directly indecomposable $LQF$-algebra. We first note that, by Proposition \ref{DIRECTID}, $Z(A) = \{0,1\}$. Moreover, by Proposition \ref{CENTRALCOVERGEN}, $e_d(x) = 0$ for each $x \not=1$. Let $x,y,z \in A$. Suppose that
$x = y$.  Therefore, by Eq.(\ref{ECMOL}), $xRy = 1$ and then, $e_d(xRy) = 1$. Thus $t(x,y,z) = z$. Now let us suppose that $x\not =y$. Therefore, by Eq.(\ref{ECMOL}),
$xRy < 1$ and then $e_d(xRy) = 0$. Thus $t(x,y,z) = x$. Hence, $t(x,y,z)$ is a discriminator term in $A$. Since ${\cal DI}({\cal LQF})$ generates ${\cal LQF}$ it is a discriminator variety.

\qed
\end{proof}

Combining Proposition \ref{DISC1} and Proposition \ref{BULMAN} we can establish the following result.

\begin{prop}\label{DISC12}
${\cal LQF}$ is an arithmetical semisimple variety. Therefore, $${\cal DI}({\cal LQF}) = {\cal SI}({\cal LQF}) = {\cal S}im({\cal LQF}).$$
\qed
\end{prop}

\section{Filters and congruences in $LQF$-algebras} \label{FILTERCONGRUENCE}

In this section we develop the filter theory for $LQF$-algebras. In order to do this, we first recall some basic results about filter theory for orthomodular lattices and Boolean algebras. Let $L$ be an orthomodular lattice. An {\it increasing set} in $L$ is a subset $F$ of $L$ such that  if $a\in F$ and $a \leq x$ then $x\in F$. An {\it $OML$-filter} (also called {\it perspective filter} \cite{KAL}) in $L$ is a subset $F\subseteq L$ satisfying the following conditions:

\begin{enumerate}
\item
$F$ is an increasing set.

\item
If $a,b \in F$ then $a\land b \in F$.

\item
If $a\in F$ and $a\sim_p b$ then $b\in F$.

\end{enumerate}

We denote by $Filt_{_{OML}}(L)$ complete lattice of $OML$-filters in $L$. A $OML$-filter in $L$ is called {\it proper} iff $F \not= L$, or equivalently, iff $0\not \in F$. A {\it maximal $OML$-filter} is a proper $OML$-filter maximal with respect to inclusion. In the particular case in which $L$ is a Boolean algebra the notion of filter, referred as  {\it $BA$-filter}, is characterized by the first two conditions  above. If $L$ is a Boolean algebra then we denote by $Filt_{_{BA}}(L)$  the complete lattice of $BA$-filter of $L$. It is well known that a $BA$-filter $F$ in the Boolean algebra $L$ is maximal iff, for each $x \in L$, $x\in F$ or $\neg x\in F$.

Let $L$ be an orthomodular lattice. If we denote by $Con_{_{OML}}(L)$ the congruence lattice of $L$ then the map
\begin{equation}\label{alpha}
Con_{_{OML}}(L) \ni \theta {\mapsto} F_\theta = \{x \in L: (x,1) \in \theta \}
\end{equation}
is a lattice order isomorphism from $Con_{_{OML}}(L)$ onto $Filt_{_{OML}}(L)$ (see {\rm \cite[\S 2 Theorem 6]{KAL}})  whose inverse is given by
\begin{equation}\label{alpha-1}
Filt_{_{OML}}(L) \ni F {\mapsto}  \theta_F = \{(x,y) \in L^2 : xRy \in F \}.
\end{equation}

\begin{definition}
{\rm Let $A$ be a $LQF$-algebra. A {\it $LQF$-filter} in $A$ is an $OML$-filter of $A$ which is closed under $\epsilon_d$.
}
\end{definition}

\begin{example}
{\rm Let $f:A \rightarrow B$ be a $LQF$-homomorphism. Clearly $Ker(f) = \{x\in A: f(x)=1\}$ is a $OML$-filter which is closed under $e_d$. Hence,  $Ker(f)$ is a $LQF$-filter.
}
\end{example}

We denote by $Filt_{_{LQF}}(A)$ the set of all $LQF$-filters in the  $LQF$-algebra $A$. Recall that $Filt_{_{LQF}}(A)$ is a closure system, hence it is also a complete lattice under the set inclusion. The notion of {\it maximal $LQF$-filter} is defined in an analogous way to the maximal $OML$-filters. By $Con_{_{LQF}}(A)$, we denote the congruence lattice of $A$.

\begin{prop} \label{FILTEQUIV}
Let $A$ be a $LQF$-algebra. Then the maps $\theta \mapsto F_\theta$ and $F \mapsto \theta_F$, defined in Eq.(\ref{alpha}) and Eq.(\ref{alpha-1}) respectively, are mutually inverse lattice order isomorphisms between $Con_{{LQF}}(A)$ and
$Filt_{_{LQF}}(A)$.
\end{prop}

\begin{proof}
Let us  suppose that $\theta \in Con_{{LQF}}(A)$. Since $\theta_F$ is a $OML$-congruence, $F_\theta = \{x\in L: (x,1)\in \theta_F \}$ is $OML$-filter. Let $x\in F_\theta $, that is, $(x,1)\in \theta$. Since $\theta$ is compatible with $\epsilon_d$ then $(\epsilon_d(x), 1 ) = (\epsilon_d(x), \epsilon_d(1)) \in \theta$. Thus, $\epsilon_d(x) \in F_\theta$ and $F_\theta \in Filt_{_{LQF}}(A)$.

For the converse, let us suppose that $F \in Filt_{_{LQF}}(A)$. We first note that $\theta_F$ is a $OML$-congruence. Thus, we only need to look at the compatibility of $\theta_F$ with respect to the binary operations $w$ and $w^*$.
Indeed: Let $(x_1, x_2) \in \theta_F$ and $(y_1, y_2) \in \theta_F$, that is, $x_1 R x_2 \in F$ and $y_1 R y_2 \in F$ respectively. Note that $\epsilon_d(x_1 R x_2) \in F$ and $\epsilon_d(y_1 R y_2) \in F$ because $F$ is closed by $\epsilon_d$. Therefore, $e_d (x_1 R x_2) \land e_d (y_1 R y_2) \in F$ because $F$ is also closed under the infimum. Since $F$ is an increasing set, by Proposition \ref{BOOLEANCOVERAUX}-4 and 5, we have that
\begin{itemize}
\item[]
$e_d (x_1 R x_2) \land e_d (y_1 R y_2) \leq w(x_1, y_1) R w(x_2, y_2) \in F$,

\item[]
$e_d (x_1 R x_2) \land e_d (y_1 R y_2) \leq w^*(x_1, y_1) R w^*(x_2, y_2) \in F$.

\end{itemize}
It proves that $(w(x_1, y_1), w(x_2, y_2)) \in \theta_F$ and $(w^*(x_1, y_1), w^*(x_2, y_2)) \in \theta_F$, so $\theta_F$ is compatible with $w$ and $w^*$. Therefore, $\theta_F \in Con_{{LQF}}(A)$.

Summarizing, by the above results, we have that
\begin{equation}\label{SubFILTEQUIV}
F \in Filt_{_{LQF}}(A) \hspace{0.2cm} \mbox{iff} \hspace{0.2cm} \theta_F \in Con_{_{LQF}}(A).
\end{equation}

Since the maps $F \mapsto \theta_F$ and $\theta \mapsto F_\theta$  are mutually inverse lattice-isomorphisms
between $Con_{_{OML}}(A)$ and $Filt_{_{OML}}(A)$ with respect to the $OML$-reduct
$\langle A, \lor, \land, \neg, 0,1 \rangle$, by Eq.(\ref{SubFILTEQUIV}), we have that $Filt_{_{LQF}}(A)$ and $Con_{_{LQF}}(A)$
are lattice-order isomorphic.

\qed
\end{proof}

Let $A$ be a $LQF$-algebra and $M\subseteq A$. Since $Filt_{_{LQF}}(A)$ is a complete lattice we define the {\it $LQF$-filter generated by $M$ } as
$$F_{_{LQF}}(M) = \bigcap \{F\in Filt_{_{LQF}}(A): M\subseteq F \}$$ that is, the smallest $LQF$-filter containing $M$. In particular, if $M=\{a\}$ then $F_{_{LQF}}(a)$ is called the {\it principal filtyer associated to $a$}.

\begin{prop}\label{LQFGENERATEDFILTER}
Let $A$ be $LQF$-algebra, $M\subseteq A$ and $a\in A$. Then:
\begin{enumerate}
\item
$F_{_{LQF}}(M) = \{x\in L: \exists x_1 \ldots x_n \in M \hspace{0.2cm} \mbox{such that} \hspace{0.2cm} e_d(x_1 \land \ldots \land x_n) \leq x  \}$.

\item
$F_{_{LQF}}(a) = [e_d(a),1]$.

\item
$F_{_{LQF}}(M)$ is a proper filter iff each finite subset $\{x_1 \ldots x_n\} \subseteq M$ satisfies  $e_d(x_1 \land \ldots \land x_n) > 0$.
\end{enumerate}

\end{prop}

\begin{proof}
 We first prove that $M_0 = \{x\in L: \exists x_1 \ldots x_n \in M \hspace{0.2cm} \mbox{such that} \hspace{0.2cm} e_d(x_1 \land \ldots \land x_n) \leq x  \}$ is a $LQF$-filter. Let us notice that  $M_0$ is an increasing set and, by Proposition \ref{BOOLEANCOVERAUX}-1 it is closed under $e_d$. Then it remains to prove that $M_0$ is closed under $\land$ and closed under perspectivity. Indeed:

To see that $M_0$ is closed under $\land$, let $x,y \in M_0$. Then $e_d(x_1 \land \ldots \land x_n) \leq x$ and $e_d(y_1 \land \ldots \land y_m) \leq y$ where $x_1, \ldots, x_n, y_1,  \ldots y_n \in M$. Thus, by Proposition \ref{BOOLEANCOVERAUX}-3, $e_d(x_1 \land \ldots \land x_n \land y_1,  \ldots y_n) = e_d(x_1 \land \ldots \land x_n) \land e_d(y_1 \land \ldots \land y_m) \leq x\land y$, so $x\land y \in M_0$.

To show that $M_0$ is closed under perspectivity, let $x\in M_0$ i.e., there exists $\{x_1 \ldots x_n\} \subseteq M$ such that $e_d(x_1 \land \ldots \land x_n) \leq x$, and $y\in A$ such that $x\sim_p y$. Then, by Proposition \ref{PERSPECTCENTER}-2, $e_d(x_1 \land \ldots \land x_n) \leq y$ because $e_d(x_1 \land \ldots \land x_n) \in Z(A)$. Thus $y\in M_0$, so $M_0$ is closed under perspectivity.

Therefore, $M_0$ is a $LQF$-filter. Let us notice that $M\subseteq M_0$ and, for each $F\in Filt_{_{LQF}}(A)$ containing $M$, we have that $M_0 \subseteq F$. Thus, $M_0$ is the smallest $LQF$-filter containing $M$ i.e., $M_0 = F_{_{LQF}}(M)$. This proves item 1. Item 2 and item 3 are handled similarly.

\qed
\end{proof}

\begin{prop}\label{BOOLEANFILTER}
Let $A$ be $LQF$-algebra and $F$ be a $LQF$-filter. Then:

\begin{enumerate}
\item
$F\cap Z(A) \in Filt_{_{BA}}(Z(A))$.

\item
$F= F_{_{LQF}}(F\cap Z(A))$.

\item
The map $Filt_{_{BA}}(Z(A)) \ni G \mapsto  F_{_{LQF}}(G)$ defines a lattice order isomorphism between $Filt_{_{BA}}(Z(A))$ and $Filt_{_{LQF}}(A)$.

\end{enumerate}
\end{prop}

\begin{proof}
1) Since $F$ is a $LQF$-filter it follows that, $F\cap Z(A)$ is closed under $\land$ and it is also an increasing set  in $Z(A)$.  Thus $F\cap Z(A) \in Filt_{_{BA}}(Z(A))$.

2) Let us notice that $F_{_{LQF}}(F\cap Z(A)) \subseteq F$ because $F\cap Z(A) \subseteq F$. To prove equality of these two sets, let $x \in F$. Then, $e_d(x) \in F\cap Z(A)$ and $e_d(x) \leq x$. Thus, $x \in F_{_{LQF}}(F\cap Z(A))$ and $F \subseteq F_{_{LQF}}(F\cap Z(A))$. Hence, $F= F_{_{LQF}}(F\cap Z(A))$.

3) We first note that $G \mapsto  F_{_{LQF}}(G)$ is an order inclusion preserving map and, by the above items, it is also surjective. Then we prove that the map $G \mapsto  F_{_{LQF}}(G)$ is injective. Let $G_1, G_2 \in Filt_{_{BA}}(Z(A))$ such that $G_1 \not= G_2$. Then, there exists $z\in Z(A)$ such that $z\in G_1$ and $z\not \in G_2$. Let us suppose that
$F_{_{LQF}}(G_1) = F_{_{LQF}}(G_2)$. Then $z \in F_{_{LQF}}(G_1)$ and $z\in F_{_{LQF}}(G_2)$. By Proposition \ref{LQFGENERATEDFILTER}-1 there exists $z_1 \ldots z_n \in G_2 \subseteq Z(A)$ such that $e_d(z_1 \land \ldots \land z_n) \leq z$. Note that $z_1 \land \ldots \land z_n \in G_2$ because $G_2$ is closed under $\land$ and, by Proposition \ref{BOOLEANCOVERAUX}-3,  $z_1 \land \ldots \land z_n = e_d(z_1) \land \ldots \land e_d(z_n) = e_d(z_1 \land \ldots \land z_n) \leq z$. Since $G_2$ is an increasing set then $z \in G_2$ which is a contradiction. Thus, $F_{_{LQF}}(G_1) \not = F_{_{LQF}}(G_2)$ and the map is injective. Hence, the map $Filt_{_{BA}}(Z(A)) \ni G \mapsto  F_{_{LQF}}(G)$ defines a lattice order isomorphism between $Filt_{_{BA}}(Z(A))$ and $Filt_{_{LQF}}(A)$.

\qed
\end{proof}

The above proposition shows that the filter theory of a  $LQF$-algebra is completely determined by the $BA$-filter theory of its center.

\begin{prop}\label{MAXLQFFILTER}
Let $A$ be $LQF$-algebra and $F$ be a $LQF$-filter. Then the following statement are equivalent

\begin{enumerate}
\item
$F$ is maximal.

\item
$F\cap Z(A)$ is a maximal $BA$-filter in $Z(A)$.

\item
For each $x\in A$, $x\in F$ or $\neg e_d(x) \in F$.

\end{enumerate}
\end{prop}

\begin{proof}
$1 \Longrightarrow 2$) Let us assume that $F$ is maximal. Let $z \in Z(A)$ such that $z\not \in F\cap Z(A)$, i.e. $z\not \in F$. Then, $A = F_{_{LQF}}(F\cup \{z\})$ because $F$ is maximal. Therefore, $0\in F_{_{LQF}}(F\cup \{z\})$ and, by Proposition \ref{LQFGENERATEDFILTER}-1, there exists $x_1 \ldots x_n \in F$ such that $e_d(x_1) \land \ldots \land e_d(x_n) \land z \leq 0$ where $e_d(x_1) \ldots e_d(x_n) \in F\cap Z(A)$. In this way,
$e_d(x_1) \land \ldots \land e_d(x_n)  \leq \neg z$ and $\neg z \in F$ because $F$ is an increasing set. Hence, $\neg z \in F\cap Z(A)$ and $F\cap Z(A)$ is a maximal $BA$-filter in $Z(A)$.

$2 \Longrightarrow 3$) Let $x\in A$ such that $x \not \in F$. Then, $e_d(x) \not \in F\cap Z(A)$ and, consequently, $\neg e_d(x) \in F\cap Z(A)$ because $F\cap Z(A)$ is a maximal $BA$-filter in $Z(A)$. Hence our claim.

$3 \Longrightarrow 1$) Let us suppose that $F$ is not maximal. Then, there exists $x\in A$ such that $x\not \in F$ and  $F_{_{LQF}}(F\cup \{x\})$ is a proper $LQF$-filter. Thus, $e_d(x) \in F_{_{LQF}}(F\cup \{x\})$ and, by hypothesis,
$\neg e_d(x) \in F \subseteq F_{_{LQF}}(F\cup \{x\})$. Consequently, $F_{_{LQF}}(F\cup \{x\})$ is not proper which is a contradiction. Hence, $F$ is maximal.

\qed
\end{proof}

\begin{prop}\label{CEPHOOP}
${\cal LQF}$ satisfies CEP.
\end{prop}

\begin{proof}
Let $A$ be a  $LQF$-algebra and let $B$ be a sub $LQF$-algebra of $A$. For each $F\in Filt_{_{LQF}}(B)$, let $F_{_{LQF}}^A(F)$ be the $LQF$-filter of $A$ generated by $F$. Clearly $F \subseteq B\cap F_{_{LQF}}^A(F)$.
To prove equality of these two sets, let $x\in B\cap F_{_{LQF}}^A(F)$. By Proposition \ref{LQFGENERATEDFILTER}-1 there exist $x_1,\cdots ,x_n \in F$ such that $e_d(x_1 \land \cdots \land x_n) \leq x$. Since $x\in B$ and $F$ is an $LQF$-filter of $B$, hence upward closed, it follows that $x\in F$, so $B \cap F_{_{LQF}}^A(F) \subseteq F$. Thus, ${\cal LQF}$ satisfies CEP.

\qed
\end{proof}

\section{A Hilbert style calculus for ${\cal LQF}$} \label{HILBERTSTYLE}

In this section we give a Hilbert-style presentation for $LQF$-logic and we prove strong completeness with respect to the variety ${\cal LQF}$.

Let $X$ be a denumerable set of variable. The language of the calculus is given by the absolutely free algebra $Term_{_{{\cal LQF}}}(X)$. In this case, valuations are homomorphisms of the form $v: Term_{_{{\cal LQF}}}(X) \rightarrow A$ where $A\in {\cal LQF}$. A term $t\in Term_{_{{\cal LQF}}}(X)$ is said to be a {\it tautology} iff for each valuation $v$, $v(t) = 1$. In this framework we regard $LQF$-terms as propositions and valid equations of the form $t=1$ as tautologies. Each subset $T \subseteq Term_{_{{\cal LQF}}}(X)$ is referred as a {\it theory}. If $v$ is a valuation then $v(T) = 1$ means that $v(t)= 1$ for each $t \in T$.
Let $t \in Term_{_{{\cal LQF}}}(X)$ and $T$ be a theory. We use $T \models_{_{\cal LQF}} t$, read $t$ is a {\it semantic consequence} of $T$, in the case in which when $v(T) = 1$ then $v(t) = 1$ for each valuation $v$.
In order to establish a Hilbert style calculus for ${\cal LQF}$ let us again consider the following notation
\begin{eqnarray*}
w_t (s) & \mbox{for} & w(t,s), \\
w^*_t (s) & \mbox{for} & w^*(t,s), \\
t R s & \mbox{for} & (t \land s) \lor (\neg t \land \neg s), \\
e_d(t)& \mbox{for} & \neg w_0(\neg t).
\end{eqnarray*}

\begin{definition}\label{OMLCALCULUS}
{\rm The calculus $\langle \mathit{Term}_{_{\cal LQF}}(X), \vdash \rangle$ is given by the
following axioms:

\begin{enumerate}[\hspace{18pt}{\rm A}1.]

\item[A0.]
$(t \lor \neg t)R1$ and $(t \land \neg t)R0 $,

\item
$t R t$,

\item
$\neg(t R s) \lor (\neg(s R r)\lor (t R
r))$,
\item
$\neg(t R s) \lor (\neg t R \neg s) $,

\item
$\neg(t R s) \lor ((t \land r) R (s \land r))$,

\item
$(t \land s) R (s \land t)$,

\item
$(t \land (s \land r)) R ((t \land s) \land r)$,

\item
$(t \land (t \lor s)) R t$,

\item
$(\neg t \land t) R ((\neg t \land t)\land s)$,

\item
$t R \neg \neg t$,

\item
$\neg(t \lor s)R(\neg t \land \neg s)$,

\item
$(t \lor (\neg t \land (t \lor s)) R (t \lor s)$,

\item
$(t R s) R (s R t)$,

\item
$\neg(t R s) \lor (\neg t \lor s)$,

\item
$w_0(0) R 0$,

\item
$x R (x \land w_0(x))$,

\item
$y R \big( (y\land w_0(x)) \lor (y\land \neg w_0(x))  \big)$,

\item
$\big( w_z(x\land y) \lor w_z(y) \big) R w_z(y)$,

\item
$\Big( \big( w_0(z) \land w_z^*(x\land y)\big) \lor w_z^*(y)\Big) R w_z^*(y)$,

\item
$\big( w_0^*(z) \land z \big) R \big( w_0^*(z) \land w_z^*(z) \big)$,

\item
$\big( w_0^*(z) \land z \big) R \Big(\neg w_0( \neg z) \land  w_0^*(z) \Big) \lor \Big(w_0(\neg z) \land w_0\big( w_0^*(z) \land z \big) \Big)$,

\item
$\big( w_0^*(z) \lor z \big) R \big( w_0^*(z) \lor w_z^*(z) \big)$,

\item
$\big( w_0^*(z) \lor z \big) R  w_0\big( w_0^*(z) \lor z \big) $,

\item
$w_0(z) R \Big( w_0(z) \land \big( w_z(w_z^* (x)) R x \big) \Big)  $,

\item
$w_0(z) R \Big( w_0(z) \land \big( w_z^*\big( w_z(x)\big) R \mu_z(x) \big)  \Big)$,

\item
$w_0\big( w_0^*(1)\big) R w_0\big(\neg w_0^*(1) \big)$,

\item
$w_x(y\land w_0(z)) R w_{x\land w_0(z)}(y\land w_0(z))$,

\item
$w_x(y\land w_0(z)) R \big( w_x(y) \land w_0(z) \big)$,

\item
$w^*_x(y\land w_0(z)) R w^*_{x\land w_0(z)}(y\land w_0(z)) $,

\item
$w^*_x(y\land w_0(z)) R \big( w^*_x(y) \land w_0(z) \big)$,

\item
$\neg e_d(tRs) \lor \big(w_r(t) R\hspace{0.1cm} w_r(s)\big)$,

\item
$\neg e_d(tRs) \lor \big(w_t(r)R\hspace{0.1cm} w_s(r)\big)$,

\item
$\neg e_d(tRs) \lor \big(w^*_r(t) R\hspace{0.1cm} w^*_r(s)\big)$,

\item
$\neg e_d(tRs) \lor \big( w^*_t(r) R\hspace{0.1cm} w^*_s(r)\big)$.

\end{enumerate}

\noindent and the following inference rules:

$$
\displaylines{ \hfill {t, \neg t \lor s \over s},
\hfill \llap{\it disjunctive syllogism (DS)} \cr\cr \hfill {t
\over e_d(t)}. \hfill \llap{\it necessitation (N)} }
$$
}
\end{definition}

Let $T$ be a theory in $\mathit{Term}_{_{\cal LQF}}(X)$.  A {\it proof}  from $T$ is a sequence
$t_1,...,t_n$ in $\mathit{Term}_{_{\cal LQF}}(X)$ such that each member is
either an axiom or a member of $T$ or follows from some preceding
member of the sequence using {\it DS} or  {\it N}. If $t \in \mathit{Term}_{_{\cal LQF}}(X)$,  $T \vdash t$
means that  $t$ is provable from $T$, that is, $t$ is the
last element of a proof from $T$. If $T = \emptyset$ then we use the
notation $\vdash t$ and in this case we will say that $t$
is a theorem of the calculus $\langle \mathit{Term}_{_{\cal LQF}}(X), \vdash \rangle$. The theory $T$ is {\it
inconsistent} if and only if $T\vdash t$ for each $t \in \mathit{Term}_{_{\cal LQF}}(X)$; otherwise it is {\it consistent}. \\

\begin{prop} \label{PRESOUDOML}

\begin{enumerate}
\item
Axioms of  $\langle \mathit{Term}_{_{\cal LQF}}(X), \vdash \rangle$ are tautologies.

\item
Let $t, s \in  \mathit{Term}_{_{\cal LQF}}(X)$ and $v$ be a valuation such that  $v(t) = v(\neg t \lor s) = 1$. Then $v(s) = 1$ i.e., $DS$ preserves $1$-valuations.

\item
Let $t \in  \mathit{Term}_{_{\cal LQF}}(X)$ and $v$ be a valuation such that  $v(t) = 1$. Then $v(e_d(t)) = 1$ i.e., $N$ preserves $1$-valuations.

\end{enumerate}
\end{prop}

\begin{proof}
1) For axioms A0 $\ldots$ A13 we refer to {\rm \cite[\S 15]{KAL}}. Since axioms A14 $\ldots$ A29 are rephrased forms of the axioms of $LQF$-algebras, by Eq.(\ref{ECMOL}), they are tautologies. Lastly,
by Proposition \ref{BOOLEANCOVERAUX}-4 and 5,  A30 $\ldots$ A33 are also tautologies,

2,3) If $v(t) = v(\neg t \lor s) = 1$ then $v(s) = 0 \lor v(s) = \neg v(t) \lor v(s) = v(\neg t \lor s) = 1$. Thus, $DS$ preserves $1$-valuations. The preservation of $1$-valuations across the inference rule $N$ is immediate.

\qed
\end{proof}

An immediate consequence of the last proposition is the following.

\begin{theo} \label{SOUDMOD} {\rm [Soundness]} Let $T$ be a theory in $\mathit{Term}_{_{\cal LQF}}(X)$ and $t \in  \mathit{Term}_{_{\cal LQF}}(X)$. Then:
$$T\vdash t \Longrightarrow T\models_{{\cal LQF}} t.$$
\qed
\end{theo}

\begin{prop}\label{COR}
Let $T$ be a theory in $\mathit{Term}_{_{\cal LQF}}(X)$ and $t, s, r \in   \mathit{Term}_{_{\cal LQF}}(X)$.
Then:

\begin{enumerate}

\item
$ \vdash t \lor \neg t$,

\item
$ \vdash 1$,

\item
$T \vdash t R s \Longrightarrow  T \vdash s  R t $,

\item
$T \vdash t R s \hspace{0.2cm} and \hspace{0.2cm} T \vdash
s R r  \hspace{0.2cm} \Longrightarrow \hspace{0.2cm} T
\vdash t  R r $,

\item
$T \vdash t R s \hspace{0.2cm} \Longrightarrow
\hspace{0.2cm} T \vdash \neg t R \neg s$,

\item
$T \vdash t R s \hspace{0.2cm} and \hspace{0.2cm} T \vdash
t \land r  \hspace{0.2cm} \Longrightarrow \hspace{0.2cm} T
\vdash s  \land r $,

\item
$T \vdash t R s \hspace{0.2cm} and \hspace{0.2cm} T \vdash
t \lor s  \hspace{0.2cm} \Longrightarrow \hspace{0.2cm} T
\vdash s  \lor r $,

\item
$T \vdash tRs \hspace{0.2cm} \Longrightarrow \hspace{0.2cm}  T \vdash w_r(t) R\hspace{0.1cm} w_r(s)$,

\item
$T \vdash tRs \hspace{0.2cm} \Longrightarrow \hspace{0.2cm} T \vdash w_t(r) R\hspace{0.1cm} w_s(r)$,

\item
$T \vdash tRs \hspace{0.2cm} \Longrightarrow \hspace{0.2cm} T \vdash w^*_r(t) R\hspace{0.1cm} w^*_r(s)$,

\item
$T \vdash tRs \hspace{0.2cm} \Longrightarrow \hspace{0.2cm} T \vdash w^*_t(r) R\hspace{0.1cm} w^*_s(r)$.

\end{enumerate}
\end{prop}

\begin{proof}
\noindent 1) It follows from A1 and A13.

\noindent 2)
\begin{enumerate}
\item[(1)]
$\vdash t \lor \neg t$ \hspace*{\fill} {\footnotesize by item 1}

\item[(2)]
$\vdash (t \lor \neg t)R1$ \hspace*{\fill} {\footnotesize by A0}

\item[(3)]
$\vdash \neg((t \lor \neg t) R 1) \lor (\neg (t \lor \neg t) \lor 1)$ \hspace*{\fill} {\footnotesize by A13}

\item[(4)]
$\vdash \neg (t \lor \neg t) \lor 1 $ \hspace*{\fill} {\footnotesize by $DS$ 2,3}

\item[(5)]
$\vdash 1 $ \hspace*{\fill} {\footnotesize by $DS$ 4,1}

\end{enumerate}

\noindent 3)
\begin{enumerate}
\item[(1)]
$T \vdash t R s$

\item[(2)]
$T \vdash (t R s) R (s R t)$ \hspace*{\fill} {\footnotesize by A12}

\item[(3)]
$T \vdash \neg ((t R s) R (s R t)) \lor (\neg
(t R s) \lor (s R t)) $ \hspace*{\fill} {\footnotesize by A13}

\item[(4)]
$T \vdash (\neg (t R s) \lor (s R t)) $
\hspace*{\fill} {\footnotesize by $DS$ 2,2}

\item[(5)]
$T \vdash s R t $ \hspace*{\fill} {\footnotesize by $DS$
1,4}

\end{enumerate}

\noindent 4) It easily follows from A2 and two application of the ${\it DS}$.

\noindent 5) It follows from A3. \hspace{0.2 cm}

\noindent 6)
\begin{enumerate}
\item[(1)]
$T \vdash t R s$

\item[(2)]
$T \vdash t \land r$

\item[(3)]
$T \vdash \neg (t R s) \lor ((t \land r)R(s
\land r))  $ \hspace*{\fill} {\footnotesize by A4}

\item[(4)]
$T \vdash (t \land r)R(s \land r)$ \hspace*{\fill} {\footnotesize by $DS$ 1,2}

\item[(5)]
$\vdash \neg ((t \land r)R(s \land r)) \lor (\neg (t
\land r) \lor (s \land r))$ \hspace*{\fill} {\footnotesize by A13}

\item[(6)]
$T \vdash s \land r$ \hspace*{\fill} {\footnotesize by $DS$
5,4,2}

\end{enumerate}

\noindent 7) It follows by item 4, A9 and A10. \hspace{0.2 cm}

\noindent 8)

\begin{enumerate}
\item[(1)]
$T \vdash t R s$

\item[(2)]
$T \vdash e_d(t R s)$ \hspace*{\fill} {\footnotesize by $N$ 1}

\item[(3)]
$\vdash \neg e_d(tRs) \lor \big(w_r(t) R\hspace{0.1cm} w_r(s)\big)$ \hspace*{\fill} {\footnotesize by A30}

\item[(4)]
$T \vdash w_r(t) R\hspace{0.1cm} w_r(s)$ \hspace*{\fill} {\footnotesize by $DS$ 2,3}

\end{enumerate}

\noindent 9,10,11) These items can be proved in an exact way as the item 8 by taking into account axioms A31, A32, A32 respectively.

\qed
\end{proof}

\begin{prop}\label{LINDENBAUMOML}
Let $T$ be a theory in $\mathit{Term}_{_{\cal LQF}}(X)$ and let us consider the binary relation in $\mathit{Term}_{_{\cal LQF}}(X)$ given by
$$t \equiv_{_T} s \hspace{0.3cm} \mbox{iff} \hspace{0.3cm} T\vdash tRs.$$ Then $\equiv_{_T}$ is an equivalence in $\mathit{Term}_{_{\cal LQF}}(X)$. Moreover if we define the following operations on ${\cal L}_T(X) = \mathit{Term}_{_{\cal LQF}}(X) /_{\equiv_T}$
\begin{align}
[t]_{_T}\land [s]_{_T} &= [t \land s]_{_T}, & \neg[t]_{_T} &=  [\neg t]_{_T}, \nonumber \\
[t]_{_T}\lor [s]_{_T} &= [t \lor s]_{_T}, & 0 &= [0]_{_T},\nonumber \\
w([t]_{_T}, [s]_{_T}) &= [w(t,s)]_{_T}, & 1 &= [1]_{_T},\nonumber \\
w^*([t]_{_T}, [s]_{_T}) &= [w^*(t,s)]_{_T}. & \nonumber
\end{align}
Then we have
\begin{enumerate}
\item
$ \langle {\cal L}_T(X), \land, \lor, w, w^*, \neg,  0, 1 \rangle$ is a $LQF$-algebra.

\item
$T\vdash t$ if and only if $[t]_{_T} = 1$.

\end{enumerate}
\end{prop}

\begin{proof}
By Axiom A1 and Proposition \ref{COR}-3 and 4, $\equiv_{_T}$ is an equivalence in $\mathit{Term}_{_{\cal LQF}}(X)$. We first prove that the operations $ \langle \land, \lor, w, w^*, \neg,  0, 1 \rangle$ are well defined on ${\cal L}_T(X)$ i.e, they are compatible operations with respect to $\equiv_{_T}$. Let us suppose that $$T \vdash t_1 R t_2 \hspace{0.3cm} \mbox{and} \hspace{0.3cm} T \vdash s_1 R s_2.$$

We prove that $T\vdash (t_1 \land t_2)R(s_1 \land s_2)$.

\begin{enumerate}
\item[(1)]
$T \vdash  t_1 R t_2$ \hspace*{\fill} {\footnotesize hypothesis}

\item[(2)]
$\vdash \neg(t_1 R t_2) \lor ((t_1 \land s_1) R (t_2 \land s_1))$ \hspace*{\fill}
{\footnotesize  A4}

\item[(3)]
$T \vdash (t_1 \land s_1) R (t_2 \land s_1)$ \hspace*{\fill}
{\footnotesize  by DS 1,2}

\item[(4)]
$T \vdash  s_1 R s_2$ \hspace*{\fill} {\footnotesize hypothesis}

\item[(5)]
$\vdash \neg(s_1 R s_2) \lor ((s_1 \land t_2) R (s_2 \land t_2))$ \hspace*{\fill}
{\footnotesize  A4}

\item[(6)]
$T \vdash (s_1 \land t_2) R (s_2 \land t_2)$ \hspace*{\fill}
{\footnotesize  by DS 4,5}

\item[(7)]
$T \vdash (t_1 \land s_1) R (t_2 \land s_2)$ \hspace*{\fill}
{\footnotesize  by 3,6, A5, Proposition \ref{COR}-4 }

\end{enumerate}

It proves that $\land$ is well defined on ${\cal L}_T(X)$. By Proposition \ref{COR}-5 we can easily show that $\neg$ is well defined on ${\cal L}_T(X)$. Next, combining  A10 with the previous results
it is easily seen that $T \vdash (t_1 \lor t_2)R(s_1 \lor s_2)$. Thus $\lor$ is also well defined on ${\cal L}_T(X)$.    \\

We prove that $T\vdash w(t_1, s_1)Rw(t_2, s_2)$.

\begin{enumerate}
\item[(1)]
$T \vdash  s_1 R s_2$ \hspace*{\fill} {\footnotesize hypothesis}

\item[(2)]
$ T \vdash w(t_1, s_1) R\hspace{0.1cm} w(t_1, s_2)$ \hspace*{\fill}
{\footnotesize  Proposition \ref{COR}-8}

\item[(3)]
$T \vdash  t_1 R t_2$ \hspace*{\fill} {\footnotesize hypothesis}

\item[(4)]
$ T \vdash w(t_1, s_2) R\hspace{0.1cm} w(t_2, s_2)$ \hspace*{\fill}
{\footnotesize  Proposition \ref{COR}-9}

\item[(5)]
$T\vdash w(t_1, s_1)Rw(t_2, s_2)$ \hspace*{\fill} {\footnotesize  by 3,5 and Proposition \ref{COR}-4}

\end{enumerate}
It proves that $w$ is well defined on ${\cal L}_T(X)$. Analogously, by using Proposition \ref{COR}-10 and 11, we can also prove that $w^*$ is well defined on ${\cal L}_T(X)$.

1) By straightforward calculation it can be seen that the reduct
$ \langle {\cal L}_T(X), \land, \lor, \neg,  0, 1 \rangle$ is an orthomodular lattice (for more details we refer to {\rm \cite[\S 15, 1. Proposition]{KAL}}).
Since Axioms A14 $\ldots$ A33 are the axioms of $LQF$-algebras rephrased in terms of the secondary connective $tRs$, we have that $ \langle {\cal L}_T(X), \land, \lor, w, w^*, \neg,  0, 1 \rangle$ is a $LQF$-algebra.

2) Let us notice that $[tR1]_{_T} = [t]_{_T}R[1]_{_T} = [t]_{_T}$. Thus, $T\vdash t$ if and only if $T\vdash tR1$ iff  $[t]_{_T} = 1$.

\qed
\end{proof}

\begin{theo} \label{COM} {\rm [Strong Completeness]} Let $t \in  \mathit{Term}_{_{\cal LQF}}(X)$ and $T$ be a theory in $\mathit{Term}_{_{\cal LQF}}(X)$.  Then,
$$T\models_{{\cal LQF}} t  \Longrightarrow    T \vdash t.  $$
\end{theo}

\begin{proof}
If $T$ is inconsistent, this result is trivial. Let us assume that  $T$ is consistent and that $T\models_{{\cal LQF}} t$. Let us suppose that $T$ does not prove $t$.
If we consider the valuation $v: \mathit{Term}_{_{\cal LQF}}(X) \rightarrow {\cal L}_T(X)$ such that $v(s) = [s]_{_T}$ then, by
Proposition \ref{LINDENBAUMOML}-2, $[t]_{_T} \neq 1$ which is a contradiction. Hence $T \vdash t$.

\qed
\end{proof}

\begin{coro} \label{COMPACTMOD} {\rm (Compactness)}
Let $t \in  \mathit{Term}_{_{\cal LQF}}(X)$ and  $T$ be a theory in $\mathit{Term}_{_{\cal LQF}}(X)$. Then,
$T\models_{{\cal LQF}} t$ iff there exists a finite subset
$T_0 \subseteq T$  such that $T_0 \models_{{\cal LQF}} t$.
\end{coro}

\begin{proof}
Let us suppose that $T\models_{{\cal LQF}} t$. Then, by Theorem \ref{COM}, $T \vdash t$ and we can suppose that $t_1, \cdots t_m, t$ is a
proof of $t$ from $T$. If we consider the finite set $T_0 = T \cap \{t_1, \cdots t_n\}$ then $T_0 \vdash t$ and, by Theorem \ref{SOUDMOD},
we have that $T_0 \models_{{\cal LQF}} t $.

\qed
\end{proof}

We can also establish a kind of deduction theorem for $\langle \mathit{Term}_{_{\cal LQF}}(X), \vdash \rangle$.

\begin{theo} \label{DED2}
Let $s, t \in \mathit{Term}_{_{\cal LQF}}(X)$ and $T$ be a theory in $\mathit{Term}_{_{\cal LQF}}(X)$. Then we
have that: $$T \cup \{s \} \vdash t \hspace{0.3cm} iff
\hspace{0.3cm} T \vdash \neg e_d(s) \lor t.$$
\end{theo}

\begin{proof}
Let us suppose $T \cup \{s \} \vdash t$. Then, by Corollary \ref{COMPACTMOD}, there exists
$t_1 \ldots t_n \in T$ such that $(t_1 \land \ldots \land t_n) \land s \models_{{\cal LQF}} t$.
Let $r = t_1 \land \ldots \land t_n$. Then $r \land s \models_{{\cal LQF}} t$ implies
that $(r \land s) \lor \neg e_d(s) \models_{{\cal LQF}} \neg e_d(s) \lor t$. Therefore, $r \lor \neg
e_d(s)  \models_{{\cal LQF}} \neg e_d(s) \lor t$
because for each valuation $v$, $v(e_d(s))$ is a central element
and $v(s \lor \neg e_d(s)) = 1$. Since $r \models_{{\cal LQF}} r \lor \neg e_d(s)$ then $r \models_{{\cal LQF}}  \neg e_d(s) \lor t$. Thus, $T \models_{{\cal LQF}}  \neg e_d(s) \lor t$ and, by  Theorem \ref{COM},  $T \vdash \neg e_d(s) \lor t$.

On the other hand, let us suppose that $T \vdash \neg e_d(s) \lor t$. Then, there exist $t_1 \ldots t_n \in T$ such that, by defining $r = t_1 \land \ldots \land t_n$,
$r \models_{{\cal LQF}}  \neg e_d(s) \lor t$. Therefore, we also have that $r \land e_d(s) \models_{{\cal LQF}} e_d(s) \land (\neg e_d(s) \lor t)$ and, consequently,
$r \land e_d(s) \models_{{\cal LQF}} e_d(s) \land t $
because for each valuation $v$, $v(e_d(s) \land (\neg e_d(s) \lor t)) = v(e_d(s) \land t)$.
Since $e_d(s) \land t \models_{{\cal LQF}} t$ we have that
\begin{equation}\label{auxDED}
r \land e_d(s) \models_{{\cal LQF}} t.
\end{equation}
Let $v$ be a valuation such that $v(T) = 1$ and $v(s) = 1$. Then, $v(r) = 1$ and $v(e_d(s)) = e_d(v(s)) = e_d(1) = 1$. Thus, by Eq.(\ref{auxDED}), $v(t) = 1$ proving that
$T \cup \{s\} \models_{{\cal LQF}} t$. Hence, by Theorem \ref{COM}, $T \cup \{s \} \vdash t$.

\qed
\end{proof}

\section*{Acknowledgments}
This work is supported by i) MIUR, project PRIN 2017: Theory and applications
of resource sensitive logics, CUP: 20173WKCM5.

\end{document}